\documentclass[10pt,a4paper,reqno]{amsart}
\usepackage{acronym}
\usepackage{amsfonts}
\usepackage{amsmath}
\usepackage{amssymb}
\usepackage{amsthm}
\usepackage{bm}
\usepackage{braket}
\usepackage{cite}
\usepackage{color}
\usepackage{geometry}
\usepackage{graphicx}
\usepackage{caption,subcaption}
\usepackage{hyperref}
\usepackage{mathrsfs}
\usepackage{mathtools}
\usepackage{setspace}
\usepackage{stmaryrd}
\usepackage{tikz}

\newgeometry{top=2cm,bottom=2cm,outer=1.5cm,inner=1.5cm}

\newcommand{\bse}{\begin{subequations}}
\newcommand{\ese}{\end{subequations}}

\newtheorem{theorem}{Theorem}[section]

\newtheorem{remark}[theorem]{Remark}

\newtheorem{proposition}[theorem]{Proposition}
\numberwithin{equation}{section}

\DeclareMathOperator{\diag}{diag}

\def\wt#1{\tilde{#1}}

\newcommand{\rT}{\mathrm{T}}

\newcommand{\cN}{\mathcal{N}}

\newcommand{\bL}{\mathbf{L}}

\newcommand{\bU}{\mathbf{U}}

\newcommand{\bM}{\mathbf{M}}
\newcommand{\bI}{\mathbf{I}}
\newcommand{\bV}{\mathbf{V}}

\newcommand{\ld}{\lambda}

\newcommand{\oa}{\omega}

\newcommand{\bSg}{\mathbf{\Sigma}}
\newcommand{\bT}{\mathbf{T}}

\acrodef{2D}[2D]{two-dimensional}
\acrodef{2DTL}[2DTL]{two-dimensional Toda lattice}
\acrodef{3D}[3D]{three-dimensional}
\acrodef{ABS}[ABS]{Adler--Bobenko--Suris}
\acrodef{bDT}[bDT]{binary Darboux transformation}
\acrodef{BSQ}[BSQ]{Boussinesq}
\acrodef{CAC}[CAC]{consistency-around-the-cube}
\acrodef{DT}[DT]{Darboux transformation}
\acrodef{DL}[DL]{direct linearisation}
\acrodef{FX}[FX]{Fordy--Xenitidis}
\acrodef{GD}[GD]{Gel'fand--Dikii}
\acrodef{KP}[KP]{Kadomtsev--Petviashvili}
\acrodef{KdV}[KdV]{Korteweg--de Vries}
\acrodef{HM}[HM]{Hirota--Miwa}
\acrodef{MDC}[MDC]{multi-dimensional consistency}
\acrodef{NQC}[NQC]{Nijhoff--Quispel--Capel}
\acrodef{ODE}[ODE]{ordinary differential equation}
\acrodef{PDE}[PDE]{partial differential equation}
\acrodef{PDeltaE}[P$\Delta$E]{partial difference equation}
\acrodef{sG}[sG]{sine--Gordon}

\title[$\mathbb{Z}_\mathcal{N}$ graded discrete integrable systems and Darboux transformations]
{$\mathbb{Z}_\mathcal{N}$ graded discrete integrable systems and Darboux transformations}

\author{Ying Shi}
\address{School of Science \\ Zhejiang University of Science and Technology \\ Hangzhou 310023 \\ China}
\email{yingshi@zust.edu.cn}

\begin{document}

\begin{abstract}
We present the Darboux transformations for a novel class of two-dimensional discrete integrable systems named as  $\mathbb{Z}_\mathcal{N}$ graded discrete integrable systems, which were firstly proposed by Fordy and Xenitidis within the framework of $\mathbb{Z}_\mathcal{N}$ graded discrete Lax pairs very recently.
In this paper, the  $\mathbb{Z}_\mathcal{N}$ graded discrete equations in coprime case  and their corresponding Lax pairs are derived from the discrete Gel'fand-Dikii hierarchy by applying a transformation of the independent variables.
The construction of the Darboux tranformations is realised by considering the associated linear problems in the bilinear formalism for the $\mathbb{Z}_\cN$ graded lattice equations.
We show that all these $\mathbb{Z}_\cN$ graded  equations share a unified solution structure in our scheme.
\end{abstract}


\maketitle

\section{Introduction}\label{S:Intro}

Discrete integrable systems have played an increasingly prominent part in mathematical physics.
A number of intriguing connections have emerged between the field of discrete integrable systems and various areas of mathematics and physics in the past two decades \cite{HJN16, DIS04}.
Discrete integrable systems include many types of equations, such as difference equations whose arguments are shifted by integer or other finite steps.
Although many processes in physics are mathematically described by differential equations (and this description reflects the smoothness of natural processes as we often experience in macroscopic phenomena), there are many physical processes (such as the ones in quantum physics) that are of an inherently discrete nature and are better described by difference equations rather than differential equations.
If we would take the continuum limit that the step size becomes infinitesimally small, then we usually recover a corresponding differential equation.
Therefore, in some sense the discrete systems are the perfect integrators for their continuous counterparts and they are widely believed to be more fundamental than their continuous versions\cite{HJN16}.
However, the difference equations before taking the limit are essentially nonlocal. This nonlocality makes such systems both be richer as well as more difficult to deal with.
Therefore, it is meaningful to well understand the nature of a discrete integrable system by developing classical methods and  inventing new mathematical tools.

 The most salient or characteristic member we shall be concerned with of the class of discrete integrable systems is the celebrated Hirota-Miwa (HM) equation (or say the discrete bilinear KP equation) \cite{Hirota1981, Miwa1982},
mainly because this discrete equation is the base member in a hierarchy which is equivalent, after a change of coordinates (Miwa tranformation\cite{Miwa1982}), to the whole continuous KP hierarchy.
It is commonly known that the continuous KP hierarchy can be reduced to all those well-known two-dimensional soliton hierarchies associated with their linear equations\cite{JM82}.
In fact,  the HM equation also plays a role as a master model in the discrete systems, since many two-dimensional discrete integrable equations such as the discrete Korteweg-de Vries (KdV) type, Boussinesq (BSQ) type, etc., can be obtained from the HM equation by taking reductions\cite{SNZ14,SNZ17,Hone17}.
This means that many things can be inherited, such as Lax pair, Darboux transformations, exact solutions, etc., from the HM equation.
Therefore our next section will be devoted to an overview of the necessary technical background material on the discrete KP-type equations.

A key feature of discrete integrable systems is the property called multi-dimensional consistency\cite{NW01, ABS03}.
This property means that a nonlinear equation can be consistently extended to equations through introducing an arbitrary number of  discrete independent variables ( together with their corresponding lattice parameters) or continuous independent variables, cf.\cite{DS97, NW01}.
Therefore, both discrete and continuous equations can be simultaneously embedded into an infinite-dimensional space spanned by both discrete and continuous coordinates.
The multi-dimensional consistency property was later employed to the classification of scalar affine-linear discrete integrable systems \cite{ABS03}.
In this classification, the fundamental equations are such as the H1 equation\cite{ABS03} (see Figure \ref{F:dKdV})
\bse\label{dKdV}
\begin{align}\label{dKdV:pu}
 (u_{n,m}-u_{n+1,m+1})(u_{n,m+1}-u_{n+1,m})=a_1^2-a_2^2,
\end{align}
which in fact is the following most well-known discrete potential KdV equation by a transformation $u\rightarrow u+a_1n+a_2m$
 \begin{align}\label{dKdV:u}
 (a_1+a_2+u_{n,m}-u_{n+1,m+1})(a_1-a_2+u_{n,m+1}-u_{n+1,m})=a_1^2-a_2^2;
\end{align}
and the H3$_{\delta=0}$ equation\cite{ABS03}
\begin{align}\label{dKdV:pv}
 a_1(v_{n,m}v_{n,m+1}+v_{n+1,m}v_{n+1,m+1})-a_2(v_{n,m}v_{n+1,m}+v_{n,m+1}v_{n+1,m+1})=0,
\end{align}
which in fact  is the following well-known discrete potential modified KdV equation up to a transformation $v\rightarrow i^{n_1+n_2}v$
\begin{align}\label{dKdV:v}
 a_1(v_{n,m}v_{n,m+1}-v_{n+1,m}v_{n+1,m+1})-a_2(v_{n,m}v_{n+1,m}-v_{n,m+1}v_{n+1,m+1})=0.
\end{align}
\ese
Here the subscripts $n$ and $m$ denote the discrete independent variables of the dependent variables $u$ and $v$, and $a_1$ and $a_2$ are their corresponding lattice parameters.
The above discrete potential KdV equation \eqref{dKdV:u} and potential modified KdV equation \eqref{dKdV:v} were systematically studied as integrable discrete equations within the direct linearisation framework together with the discrete Schwarzian KdV equation (i.e. the cross-ratio equation) by Nijhoff, Quispel and Capel et al., see \cite{NC95}.
\begin{figure}
\centering
\begin{tikzpicture}[scale=2]
 \coordinate (00) at (0, 0);
 \coordinate (10) at (1, 0);
 \coordinate (20) at (2, 0);
 \coordinate (01) at (0, 1);
 \coordinate (11) at (1, 1);
 \coordinate (21) at (2, 1);
 \coordinate (02) at (0, 2);
 \coordinate (12) at (1, 2);
 \coordinate (22) at (2, 2);
 \draw[very thick] (00) -- (20);
 \draw[very thick] (01) -- (21);
 \draw[very thick] (02) -- (22);
 \draw[very thick] (00) -- (02);
 \draw[very thick] (10) -- (12);
 \draw[very thick] (20) -- (22);
 \fill [black] (00) circle (2pt);
 \fill [black] (10) circle (2pt);
 \fill [black] (01) circle (2pt);
 \fill [black] (11) circle (2pt);
 \node at (0.1,-0.2) {$u_{n,m}$};
 \node at (1.2,-0.2) {$u_{n+1,m}$};
 \node at (2.2,-0.2) {$u_{n+2,m}$};
 \node at (0.3,0.8) {$u_{n,m+1}$};
 \node at (1.4,0.8) {$u_{n+1,m+1}$};
 \node at (2.4,0.8) {$u_{n+2,m+1}$};
 \node at (0.1,2.2) {$u_{n,m+2}$};
 \node at (1.2,2.2) {$u_{n+1,m+2}$};
 \node at (2.2,2.2) {$u_{n+2,m+2}$};
\end{tikzpicture}
\caption{Discrete KdV-type equations}
\label{F:dKdV}
\end{figure}
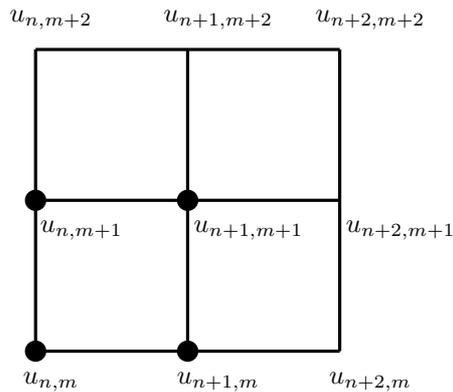
Apart from the above one-component (scalar) equations, some results do exist for multi-component discrete integrable systems, such as the Gel'fand-Dikii hierarchy which was proposed in \cite{NPCQ92}, the extended discrete BSQ-type equations in multi-component form\cite{ZZN12}, etc.
Relating to the property of multi-dimensional consistency,  a search for multi-component discrete integrable systems was also made by Hietarinta \cite{Hie11}, resulting in a remarkable classification of multi-component discrete BSQ-type equations.

Very recently, the classification of discrete integrable systems for multi-component discrete systems, associated with $\mathbb{Z}_\mathcal{N}$-graded Lax pairs ( $\mathcal{N}=2,3,\dots$), was proposed by Fordy and Xenitidis \cite{FX17} and gives rise to a large class of discrete integrable systems falling into two categories: coprime (indecomposable) and non-coprime (decomposable) cases.
In their classification, it includs some novel examples and some well known examples (such as for $\mathcal{N}=2$ having discrete potential unmodified, modified  and Schwarzian KdV equations, the discrete sine-Gordon equation; for $\mathcal{N}=3$ having the discrete potential unmodified, modified BSQ equations).
The key point in their construction is an $\cN\times\cN$ periodic matrix $\bSg$ which has very nice properties such as $\bSg^\cN=\bI$ and $\bSg^\rT=\bSg^{-1}$, where $\bI$ is the $\cN\times\cN$ unit matrix.

However, the search of Fordy and Xenitidis's  $\mathbb{Z}_\mathcal{N}$-graded discrete integrable systems is still rather primitive, especially regarding the problem of exact solutions (such as multi-soliton solutions).
Our motivation is to provide a systematic framework for exact solutions to the $\mathbb{Z}_\mathcal{N}$ graded discrete integrable systems in coprime case that were first proposed by Fordy and Xenitidis in \cite{FX17}.
The method we adopt is the \ac{DT}.
Compared with other approaches, the \ac{DT} method provides a very direct way for solving an integrable equation
since it only relies on a seed solution and a Lax pair and has been successfully applied to many nonlinear integrable \ac{PDE}s, see monographs \cite{Mat91,Rog02}.
As a compound type of a \ac{DT}, the \ac{bDT} method requires considering a Lax pair and its adjoint simultaneously
when constructing exact solutions for a nonlinear equations, providing more rich solutions, see e.g. \cite{WTLS98, NW97}.
The \ac{DT} and \ac{bDT} methods for discrete integrable systems still need to be developed,
and are so far applied to very limited examples,
including the discrete \ac{KP}-type equations equation (see \cite{Nim97,WTS97,DN09})
and the discrete unmodified and modified \ac{KdV} equations, i.e. \eqref{dKdV:u} and \eqref{dKdV:v}, cf. \cite{SNZ14,SNZ17}.

There are two issues in terms of constructing \ac{DT}s for the $\mathbb{Z}_\mathcal{N}$-graded discrete integrable equations as follows:
(1) Nonlinear potentials often appear in a nonlocal way,
which results in the fact that sometimes it is difficult to write down explicit \ac{DT} formulae;
(2) Although Lax pairs for the discrete equations are given in \cite{FX17},
one still needs their adjoint Lax pairs in order to construct the \ac{bDT}s.
In the present paper, we provide ways to both problems.
For the former one, we introduce an extra continuous variable $x$ which is associated with continuous spectral problems
compatible with the $\mathbb{Z}_\cN$ graded discrete Lax pairs,
and thus the explicit \ac{DT} formalism will rely on $x$;
while for the latter one, we give out adjoint $\mathbb{Z}_\cN$ graded discrete Lax pairs.
Our idea is to link the $\mathbb{Z}_\mathcal{N}$ graded discrete integrable models to the theory of the discrete and continuous \ac{KP} hierarchies.
By performing periodic reductions of the \ac{KP}-type equations
and considering deformations of the reduced equations,
we obtain not only the $\mathbb{Z}_\mathcal{N}$ graded discrete integrable equations (the coprime case of \ac{FX} discrete systems), the corresponding $\mathbb{Z}_\cN$ graded Lax pairs and novel adjoint discrete Lax pairs, but also their compatible continuous/semi-discrete analogues.
More importantly, such an approach also gives rise to the bilinear formalism of these equations,
which relates everything to a single key variable, i.e. the tau function.
The \ac{DT} will be constructed on the level of the tau function,
which then naturally induces those of the discrete systems in other nonlinear variables
(including the additive (i.e. unmodified) potential and the quotient (i.e. modified) potential).
As examples, we give the explicit formulae for multi-soliton solutions for the \ac{FX} equations in coprime case.

The paper is organised as follows.
We briefly review the main results in the discrete KP-type equations in section \ref{S:dKP}.
In section \ref{S:ZN}, by performing the periodic reductions on the discrete KP-type equations,
we recover the discrete Gel'fand-Dikii hierarchy, which later, by a variable transformation,  successfully realises the $\mathbb{Z}_\cN$ graded discrete integrable systems,
including the novel bilinear formalism of the \ac{FX} models,
the $\mathbb{Z}_\cN$ graded discrete$/$continuous spectral problems
and the resulting  $\mathbb{Z}_\cN$ graded semi-discrete equations.
A unified \ac{DT} for the $\mathbb{Z}_\cN$ graded discrete equations is presented in section \ref{S:DT}.
Section \ref{S:Sol} is concerned with exact solutions obtained from the \ac{DT}.

\section{Discrete Kadomtsev--Petviashvili family}\label{S:dKP}

\subsection{Nonlinear difference equations and Lax triplets}

Equations in the discrete \ac{KP} family, as \ac{3D} integrable discrete models,
are often considered as the most fundamental models in the theory of discrete integrable systems.
Below we briefly review the main results in the discrete \ac{KP} family.

There are several nonlinear equations in the discrete \ac{KP} family.
Here we only list those equations which will be used in the later sections\footnote{
Apart from the nonlinear forms given here, there also exist other nonlinear forms in the discrete \ac{KP} family
such as the discrete Schwarzian \ac{KP} equation \cite{DN91}, the (2+1)-dimensional \ac{NQC} equation \cite{NCWQ84},
and the various nonpotential versions of the discrete \ac{KP} equations, see \cite{Nim06,GRPSW07,Fu17a}.
}
as follows:
\bse\label{dKP}
\begin{align}
 &\frac{a_1-a_3+u_{n,m+1,h+1}-u_{n+1,m+1,h}}{a_1-a_3+u_{n,m,h+1}-u_{n+1,m,h}}
 =\frac{a_2-a_3+u_{n+1,m,h+1}-u_{n+1,m+1,h}}{a_2-a_3+u_{n,m,h+1}-u_{n,m+1,h}}
 =\frac{a_1-a_2+u_{n,m+1,h+1}-u_{n+1,m,h+1}}{a_1-a_2+u_{n,m+1,h}-u_{n+1,m,h}}, \label{dKP:u} \\
 &a_1\left(\frac{v_{n+1,m+1,h}}{v_{n,m+1,h}}-\frac{v_{n+1,m,h+1}}{v_{n,m,h+1}}\right)
 +a_2\left(\frac{v_{n,m+1,h+1}}{v_{n,m,h+1}}-\frac{v_{n+1,m+1,h}}{v_{n+1,m,h}}\right)
 +a_3\left(\frac{v_{n+1,m,h+1}}{v_{n+1,m,h}}-\frac{v_{n,m+1,h+1}}{v_{n,m+1,h}}\right)=0, \label{dKP:v} \\
 &(a_1-a_2)\tau_{n,m,h+1}\tau_{n+1,m+1,h}+(a_2-a_3)\tau_{n+1,m,h}\tau_{n,m+1,h+1}+(a_3-a_1)\tau_{n,m+1,h}\tau_{n+1,m,h+1}=0, \label{dKP:tau}
\end{align}
\ese
Equation for $u$ and $v$ are referred to as the (unmodified) discrete \ac{KP} equation and the discrete modified \ac{KP} equation, respectively,
which were given by Nijhoff et al. \cite{NCW85} within the direct linearisation framework.
The equation expressed by the tau function $\tau$ is often known as the \ac{HM} equation or the discrete bilinear \ac{KP} equation.
It first appeared in Hirota's paper \cite{Hirota1981} in a slightly different form
and was denoted as the discrete analogue of a generalised Toda equation.
The parametrisation in \eqref{dKP:tau} was attributed to Miwa \cite{Miwa1982},
who showed that the \ac{HM} equation actually encodes the whole hierarchy of the continuous bilinear \ac{KP} equations.

The nonlinear equations in \eqref{dKP} are actually connected with each other through the discrete Miura maps (see e.g. \cite{HJN16,Fu17a})
\bse\label{dKP:MT}
\begin{align}
 &a_1-a_2+u_{n,m+1,h}-u_{n+1,m,h}
 =a_1\frac{v_{n+1,m,h}}{v_{n,m,h}}-a_2\frac{v_{n,m+1,h}}{v_{n,m,h}}
 =(a_1-a_2)\frac{\tau_{n,m,h}\tau_{n+1,m+1,h}}{\tau_{n,m+1,h}\tau_{n+1,m,h}}, \\
 &a_2-a_3+u_{n,m,h+1}-u_{n,m+1,h}
 =a_2\frac{v_{n,m+1,h}}{v_{n,m,h}}-a_3\frac{v_{n,m,h+1}}{v_{n,m,h}}
 =(a_2-a_3)\frac{\tau_{n,m,h}\tau_{n,m+1,h+1}}{\tau_{n,m,h+1}\tau_{n,m+1,h}}, \\
 &a_1-a_3+u_{n,m,h+1}-u_{n+1,m,h}
 =a_1\frac{v_{n+1,m,h}}{v_{n,m,h}}-a_3\frac{v_{n,m,h+1}}{v_{n,m,h}}
 =(a_1-a_3)\frac{\tau_{n,m,h}\tau_{n+1,m,h+1}}{\tau_{n,m,h+1}\tau_{n+1,m,h}}.
\end{align}
\ese
We point out that on the discrete level Miura maps often take the form of nonlocal difference transforms.
This is unlike the continuous case, in which the unmodified variable $u$ itself is normally uniquely determined
by the modified variable $v$ and the tau function $\tau$.

The associated linear equations (i.e. the Lax triplet) for the discrete \ac{KP} equation \eqref{dKP:u} are given by
\bse\label{dKP:Lax}
\begin{align}
 &\phi_{n+1,m,h}-\phi_{n,m+1,h}=(a_1-a_2+u_{n,m+1,h}-u_{n+1,m,h})\phi_{n,m,h}, \label{dKP:Lax1} \\
 &\phi_{n,m+1,h}-\phi_{n,m,h+1}=(a_2-a_3+u_{n,m,h+1}-u_{n,m+1,h})\phi_{n,m,h}, \label{dKP:Lax2} \\
 &\phi_{n+1,m,h}-\phi_{n,m,h+1}=(a_1-a_3+u_{n,m,h+1}-u_{n+1,m,h})\phi_{n,m,h}, \label{dKP:Lax3}
\end{align}
\ese
see e.g. \cite{SNZ14,SNZ17}.
The corresponding Lax triplets for the modified \ac{KP} equation \eqref{dKP:v} and the \ac{HM} equation \eqref{dKP:tau}
can be obtained by replacing $u$ by $v$ and $\tau$, respectively, with the help of the Miura maps \eqref{dKP:MT}.
The Lax pair triplet \eqref{dKP:Lax} has its adjoint (cf. \cite{SNZ14,SNZ17}) which is composed of
\bse\label{dKP:Adjoint}
\begin{align}
 &\psi_{n-1,m,h}-\psi_{n,m-1,h}=(a_1-a_2+u_{n-1,m,h}-u_{n,m-1,h})\psi_{n,m,h}, \label{dKP:Adjoint1} \\
 &\psi_{n,m-1,h}-\psi_{n,m,h-1}=(a_2-a_3+u_{n,m-1,h}-u_{n,m,h-1})\psi_{n,m,h}, \label{dKP:Adjoint2} \\
 &\psi_{n-1,m,h}-\psi_{n,m,h-1}=(a_1-a_3+u_{n-1,m,h}-u_{n,m,h-1})\psi_{n,m,h}. \label{dKP:Adjoint3}
\end{align}
\ese
The compatibility of either \eqref{dKP:Lax} or \eqref{dKP:Adjoint} will provide the discrete KP-type equations in \eqref{dKP}.
The adjoint linear system here is mentioned, because we know that there are two independent spectral variables needed to completely describe the solution space of a \ac{3D} integrable systems,
and each of them governs an associate linear system.

\subsection{Associated semi-discrete systems}

The benefit of considering Miura maps between nonlinear equations in the same class
is that one can construct a solution of a nonlinear equation from that of another equation connected by the Miura map.
However, on the discrete level we have observed that the Miura maps \eqref{dKP:MT} are nonlocal difference transforms,
leading to the difficulty of constructing exact solutions directly.
For instance, assume that the solution for the \ac{HM} equation $\tau$ is given.
To construct solutions for \eqref{dKP:u} and \eqref{dKP:v},
we still have to solve the set of linear difference equations for $u$ and $v$ given in \eqref{dKP:MT}, respectively.

A natural question would be whether such a difficulty can be overcome?
The answer is positive. But instead we have to introduce extra continuous and discrete independent variables respectively. To be more precise, for the unmodified case,
we introduce a continuous variable $x$ corresponding to the lowest order flow in the continuous \ac{KP} hierarchy.
The continuous variable $x$ is compatible with the discrete variables $n$, $m$ and $h$.
This is guaranteed by the \ac{MDC} property of the discrete and continuous \ac{KP} hierarchy,
namely the discrete \ac{KP} equation can be consistently embedded into an infinite-dimensional space spanned by infinitely many discrete and continuous variables, see \cite{Fu17b}.
For the modified case, we introduce a discrete variable $k$, which is also compatible with the discrete variables $n$, $m$ and $h$ by the \ac{MDC} property of the discrete \ac{KP} hierarchy.

Without loss of generality we consider the linear equation involving $x$ and an $h$ shift for the wave function
\begin{align}\label{sdKP:Lax}
 \partial_x\phi_{n,m,h}=(u_{n,m,h+1}-u_{n,m,h}-a_3)\phi_{n,m,h}+\phi_{n,m,h+1},
\end{align}
together with the Miura maps involving $x$ and the $h$ shift between $u$, $v$ and $\tau$ given by
\begin{align}\label{sdKP:MT}
 u_{n,m,h+1}-u_{n,m,h}-a_3=\partial_x\ln v_{n,m,h}-a_3\frac{v_{n,m,h+1}}{v_{n,m,h}}
 =\partial_x\ln\tau_{n,m,h+1}-\partial_x{\ln\tau_{n,m,h}}-a_3,
\end{align}
where $\partial_x$ denotes the derivative with respect to the continuous variable $x$.
Equation \eqref{sdKP:Lax} plays the role of a linear equation in the Lax pair of the semi-discrete \ac{KP} hierarchy.
The compatibility condition between \eqref{sdKP:Lax} and \eqref{dKP:Lax3} yields
\bse\label{sdKP}
\begin{align}
 &\partial_x(u_{n,m,h+1}-u_{n+1,m,h})=(a_1-a_3+u_{n,m,h+1}-u_{n+1,m,h})(u_{n,m,h}-u_{n,m,h+1}-u_{n+1,m,h}+u_{n+1,m,h+1}), \label{sdKP:u} \\
 &\partial_x(\ln v_{n,m,h+1}-\ln v_{n+1,m,h})
 =a_1\left(\frac{v_{n+1,m,h+1}}{v_{n,m,h+1}}-\frac{v_{n+1,m,h}}{v_{n,m,h}}\right)
 -a_3\left(\frac{v_{n+1,m,h+1}}{v_{n+1,m,h}}-\frac{v_{n,m,h+1}}{v_{n,m,h}}\right), \label{sdKP:v} \\
 &\tau_{n+1,m,h}\partial_x\tau_{n,m,h+1}-\tau_{n,m,h+1}\partial_x\tau_{n+1,m,h}
 =(a_1-a_3)(\tau_{n,m,h}\tau_{n+1,m,h+1}-\tau_{n,m,h+1}\tau_{n+1,m,h}). \label{sdKP:tau}
\end{align}
\ese
These are the semi-discrete \ac{KP}-type equations with two discrete variables $n$ and $h$ and one continuous variable $x$.

Similarly, we introduce the adjoint of the linear equation \eqref{sdKP:Lax} which is in the form
\begin{align}\label{sdKP:Adjoint}
 \partial_x\psi_{n,m,h}=(u_{n,m,h-1}-u_{n,m,h}+a_3)\psi_{n,m,h}-\psi_{n,m,h-1}.
\end{align}
The compatibility between \eqref{sdKP:Adjoint} and \eqref{dKP:Adjoint3} gives rise to the same semi-discrete \ac{KP}-type equations \eqref{sdKP}.

A remark here is that although we consider the linear equation and its adjoint involving $x$ and the $h$ shift,
it is also allowed to consider the linear equations involving $x$ and either an $n$ or $m$ shift.
The resulting semi-discrete equations can also be seen as the semi-discrete KP-type equations.

One can directly verify that these semi-discrete equations, namely \eqref{sdKP:u}, \eqref{sdKP:v} and \eqref{sdKP:tau},
are compatible with the fully discrete equations \eqref{dKP:u}, \eqref{dKP:v} and \eqref{dKP:tau}, respectively.
In other words, the semi-discrete equations are multidimensionally consistent with the discrete ones.

The semi-discrete Miura map \eqref{sdKP:MT} implies that
the unmodified variable $u$ itself can be expressed by the tau function $\tau$ via a logarithm derivative.
Furthermore, we can also observe that in the \emph{zero parameter} case,
the modified variable $v$ is expressed by
the ratio of the shifted tau function and the unshifted tau function with regard to the shifted direction in zero parameter.
By consistently introducing two extra independent variables,
i.e. the continuous variable $x$ and the discrete variable $k$ associated with lattice parameter $a_4$ (as is guaranteed by the \ac{MDC} of the discrete \ac{KP}), we can conclude that
\begin{align}\label{cKP:BLT}
 u_{n,m,h}=\partial_x\ln\tau_{n,m,h} \quad \hbox{and} \quad v_{n,m,h,k}=\left.\frac{\tau_{n,m,h,k+1}}{\tau_{n,m,h,k}}\right|_{a_4=0}.
\end{align}
Such relations can also be verified in the \ac{DL} framework (see \cite{Fu18a,Fu18b}),
and also coincide with the bilinear transforms in the \ac{KP} theory established by the Sato school, cf. e.g. \cite{JM83}.

We can see that in \eqref{cKP:BLT} the dynamical evolutions with respect to $x$ and $k$ do not appear in \eqref{dKP}.
In other words, once the bilinear equation \eqref{dKP:tau} is solved,
we can immediately construct the corresponding solutions for \eqref{dKP:u} and \eqref{dKP:v} with the help of \eqref{cKP:BLT},
instead of solving \eqref{dKP:MT}.
However the price one has to pay is that the semi-discrete equation \eqref{sdKP:tau} must be considered simultaneously
when we solve the \ac{HM} \eqref{dKP:tau}, in order to add the continuous dynamical evolution in the tau function.

\section{$\mathbb{Z}_\cN$ graded discrete Lax pairs and related nonlinear equations}\label{S:ZN}

\subsection{Discrete Gel'fand--Dikii hierarchy}

We now consider the periodic reduction of the discrete \ac{KP}-type equations given in \eqref{dKP} with the purpose of obtaining the discrete Gel'fand-Dikii hierarchy.  We may introduce the following periodicity condition on the $\tau$  function
\begin{subequations}\label{Period}
\begin{align}\label{Period:NL-tau}
 \tau_{n,m,h+\cN}=\tau_{n,m,h},
\end{align}
and from \eqref{cKP:BLT} consequently obtain, with setting $a_3=0$,
\begin{align}\label{Period:NL}
u_{n,m,h+\cN}=u_{n,m,h}, \quad v_{n,m,h+\cN}=v_{n,m,h},
\end{align}
\end{subequations}
where the positive integer $\cN\geqslant2$. These are the constraints in the periodic reduction for the discrete \ac{KP} equations \eqref{dKP}.
Moreover, notice that the wave functions are fully determined by the tau function. We may also introduce the quasi-periodicity conditions of $\phi$ and $\psi$, respectively, which are given as follows:
\begin{align}\label{Period:L}
 \phi_{n,m,h+\cN}=\ld^\cN\phi_{n,m,h}, \quad \psi_{n,m,h+\cN}=\ld^{-\cN}\psi_{n,m,h}.
\end{align}

Now we can impose the constraints \eqref{Period:NL-tau} and \eqref{Period:NL}, with setting zero parameter $a_3=0$, on the discrete \ac{KP}-type equations \eqref{dKP} for obtaining their periodic reductions.
In \eqref{Period} and \eqref{Period:L}, we have seen that the periodic reduction is performed on the independent variable $h$.
Thus, the reduced equations are \ac{2D} discrete models having $n$ and $m$ as the independent variables,
and the variable $h$ plays a role as the index for each component.
For better presentation, we adopt the superscript $(\cdot)^{(h)}$ instead of the subscript $(\cdot)_h$ in the nonlinear variables in order to distinguish
the index of the components $h$ and the lattice variables $n$ and $m$. For example, $u_{n,m}^{(h)}\doteq u_{n,m,h}$.

By setting $a_3=0$, the equations in \eqref{dKP} turn out to be
\bse\label{dGD}
\begin{align}
 &\frac{a_1+u_{n,m+1}^{(h+1)}-u_{n+1,m+1}^{(h)}}{a_1+u_{n,m}^{(h+1)}-u_{n+1,m}^{(h)}}
 =\frac{a_2+u_{n+1,m}^{(h+1)}-u_{n+1,m+1}^{(h)}}{a_2+u_{n,m}^{(h+1)}-u_{n,m+1}^{(h)}}, \label{dGD:u} \\
 &a_1\left(\frac{v_{n+1,m+1}^{(h)}}{v_{n,m+1}^{(h)}}-\frac{v_{n+1,m}^{(h+1)}}{v_{n,m}^{(h+1)}}\right)
 =a_2\left(\frac{v_{n+1,m+1}^{(h)}}{v_{n+1,m}^{(h)}}-\frac{v_{n,m+1}^{(h+1)}}{v_{n,m}^{(h+1)}}\right), \label{dGD:v} \\
 &a_1\left(\tau_{n,m}^{(h+1)}\tau_{n+1,m+1}^{(h)}-\tau_{n,m+1}^{(h)}\tau_{n+1,m}^{(h+1)}\right)
 =a_2\left(\tau_{n,m}^{(h+1)}\tau_{n+1,m+1}^{(h)}-\tau_{n+1,m}^{(h)}\tau_{n,m+1}^{(h+1)}\right), \label{dGD:tau}
\end{align}
\ese
and they together with the periodicity conditions following from \eqref{Period:NL} and \eqref{Period:NL-tau},
i.e. $u_{n,m}^{(h+\cN)}=u_{n,m}^{(h)}$, $v_{n,m}^{(h+\cN)}=v_{n,m}^{(h)}$ and $\tau_{n,m}^{(h+\cN)}=\tau_{n,m}^{(h)}$,
form the equations of rank $\cN $ in the discrete unmodified, modified and bilinear \ac{GD} hierarchies, respectively.
Equations in \eqref{dGD} should be understood as $\cN$-component coupled systems.
Without losing generality, these coupled systems are composed of their respective components for $h=0,1,\cdots,\cN-1$.

We can also derive the Miura maps between the equations in \eqref{dGD}
\begin{align}\label{dGD:MT}
 a_1+u_{n,m}^{(h+1)}-u_{n+1,m}^{(h)}=a_1\frac{v_{n+1,m}^{(h)}}{v_{n,m}^{(h)}}
 =a_1\frac{\tau_{n+1,m}^{(h+1)}\tau_{n,m}^{(h)}}{\tau_{n+1,m}^{(h)}\tau_{n,m}^{(h+1)}}, \quad
 a_2+u_{n,m}^{(h+1)}-u_{n,m+1}^{(h)}=a_2\frac{v_{n,m+1}^{(h)}}{v_{n,m}^{(h)}}
 =a_2\frac{\tau_{n,m+1}^{(h+1)}\tau_{n,m}^{(h)}}{\tau_{n,m+1}^{(h)}\tau_{n,m}^{(h+1)}} ,
\end{align}
which allows that we may have $v_{n,m}^{(h)}=\tau_{n,m}^{(h+1)}/\tau_{n,m}^{(h)}$.
These difference transforms are natural consequences of the Miura maps in \eqref{dKP:MT} under the periodicity conditions \eqref{Period}.
Making use of the periodicity condition \eqref{Period:NL-tau} of the tau function, we can from \eqref{dGD:MT} further derive
\begin{align}\label{dGD:FirstIntegral}
 \prod_{h=0}^{\cN-1}v_{n,m}^{(h)}=1, \quad \hbox{and then} \quad
 \prod_{h=0}^{\cN-1}\left(a_1+u_{n,m}^{(h+1)}-u_{n+1,m}^{(h)}\right)=a_1^\cN, \quad
 \prod_{h=0}^{\cN-1}\left(a_2+u_{n,m}^{(h+1)}-u_{n,m+1}^{(h)}\right)=a_2^\cN.
\end{align}

We can also impose the periodicity conditions \eqref{Period} and \eqref{Period:L} on the Lax triplets of the discrete \ac{KP}-type equations.
By taking $a_3=0$ in \eqref{dKP:Lax} and introducing the reduced wave function having components $\phi_{n,m}^{(h)}\doteq \ld^{-h}\phi_{n,m,h}$,
we obtain linear equations as follows:
\begin{align}\label{dGD:Lax}
 \phi_{n+1,m}^{(h)}=\left(a_1+u_{n,m}^{(h+1)}-u_{n+1,m}^{(h)}\right)\phi_{n,m}^{(h)}+\ld\phi_{n,m}^{(h+1)}, \quad
 \phi_{n,m+1}^{(h)}=\left(a_2+u_{n,m}^{(h+1)}-u_{n,m+1}^{(h)}\right)\phi_{n,m}^{(h)}+\ld\phi_{n,m}^{(h+1)}
\end{align}
for $h=0,1,\cdots,\cN-1$, in which the wave function satisfies $\phi_{n,m}^{(h+\cN)}=\phi_{n,m}^{(h)}$ as it follows from \eqref{Period:L}.
The linear equations in \eqref{dGD:Lax} form the Lax pair for equation \eqref{dGD:u}.
Replacing the variable $u$ by $v$ and $\tau$, respectively, with the help of the Miura maps \eqref{dGD:MT},
we derive the respective Lax pairs for equations \eqref{dGD:v} and \eqref{dGD:tau}.

Similarly, the periodic reduction of \eqref{dKP:Adjoint} gives rise to the adjoint Lax pair for the equations in \eqref{dGD}.
Introducing $\psi_{n,m}^{(h)}\dot=\ld^h\psi_{n,m,h}$, we can write the reduced adjoint linear problem as
\begin{align}\label{dGD:Adjoint}
 \psi_{n-1,m}^{(h)}=\left(a_1+u_{n-1,m}^{(h)}-u_{n,m}^{(h-1)}\right)\psi_{n,m}^{(h)}+\ld\psi_{n,m}^{(h-1)}, \quad
 \psi_{n,m-1}^{(h)}=\left(a_2+u_{n,m-1}^{(h)}-u_{n,m}^{(h-1)}\right)\psi_{n,m}^{(h)}+\ld\psi_{n,m}^{(h-1)}
\end{align}
for $h=0,1,\cdots,\cN-1$, where the potential $u$ still obeys \eqref{dGD:u},
and each component in the wave function satisfies $\psi_{n,m}^{(h+\cN)}=\psi_{n,m}^{(h)}$ as a result of the second relation in \eqref{Period:L}.
The adjoint Lax pair for the modified equation \eqref{dGD:v} and the bilinear equation \eqref{dGD:tau} are derived with the aid of the same Miura maps \eqref{dGD:MT}.


\subsection{Deformations and $\mathbb{Z}_\cN$ graded discrete integrable models}

Reference \cite{FX17} shows that a large class of \ac{2D} integrable difference equations, including a number of novel models, were studied by Fordy and Xenitidies within the framework of $\mathbb{Z}_\cN$ graded discrete Lax pairs.
In our paper here we show that the models arising from the coprime case in their construction,
i.e. the equations corresponding to the additive and quotient potentials
in the equivalent classes $(\alpha,\alpha+1;\beta,\beta+1)$ for $\alpha,\beta=0,1,\cdots,\cN-1$,
are deformations of the discrete unmodified and modified Gel'fand-Dikii hierarchies given by \eqref{dGD:u} and \eqref{dGD:v}.
Meanwhile, by using the same deformations, we can recover their $\mathbb{Z}_\cN$ graded discrete Lax pairs from \eqref{dGD:Lax}.
Furthermore, using our point of view we provide the associated bilinear form of these unmodified and modified $\mathbb{Z}_\cN$ graded discrete integrable models.

We introduce the change of variables
\begin{align}\label{dFX:CoV}
 n=n_\alpha, \quad m=m_\beta, \quad h=l+\alpha n_\alpha+\beta m_\beta,
\end{align}
where $\alpha,\beta\in\{0,1,\cdots,\cN-1\}$.

Taking the periodicity condition \eqref{Period:NL-tau} into account, we observe that under the variable transform \eqref{dFX:CoV} the tau function satisfies
\begin{align}\label{dFX:tauTransform}
 \tau_{n_\alpha+1,m_\beta}^{(l)}=\tau_{n+1,m}^{(h+\alpha)}, \quad
 \tau_{n_\alpha,m_\beta+1}^{(l)}=\tau_{n,m+1}^{(h+\beta)} \quad \hbox{and} \quad
 \tau_{n_\alpha,m_\beta}^{(l+1)}=\tau_{n,m}^{(h+1)}.
\end{align}
If we shift $h$ in \eqref{dGD:tau} by $\alpha+\beta$ units and make use of \eqref{dFX:tauTransform},
the following $\cN$-component coupled system of bilinear equations is derived:
\bse\label{dFX}
\begin{align}\label{dFX:tau1}
 a_1\left(\tau_{n_\alpha,m_\beta}^{(l+1+\alpha+\beta)}\tau_{n_\alpha+1,m_\beta+1}^{(l)}-\tau_{n_\alpha,m_\beta+1}^{(l+\alpha)}\tau_{n_\alpha+1,m_\beta}^{(l+1+\beta)}\right)
 =a_2\left(\tau_{n_\alpha,m_\beta}^{(l+1+\alpha+\beta)}\tau_{n_\alpha+1,m_\beta+1}^{(l)}-\tau_{n_\alpha+1,m_\beta}^{(l+\beta)}\tau_{n_\alpha,m_\beta+1}^{(l+1+\alpha)}\right),
\end{align}
where $l=0,1,\cdots,\cN-1$, and $\tau_{n_\alpha,m_\beta}^{(l+\cN)}=\tau_{n_\alpha,m_\beta}^{(l)}$.
This formula shows the multi-component systems expressed by the bilinear potential,
which plays the role of the bilinear formalism of the $\mathbb{Z}_\cN$ graded discrete integrable models.
We can also prove that the relation \eqref{dFX:tauTransform} also holds for the variables $u$ and $v$.
Therefore, from equations \eqref{dGD:u} and \eqref{dGD:v} we obtain unmodified $\mathbb{Z}_\cN$ graded discrete integrable systems
\begin{align}\label{dFX:u}
 \frac{a_1+u_{n_\alpha,m_\beta+1}^{(l+1+\alpha)}-u_{n_\alpha+1,m_\beta+1}^{(l)}}{a_1+u_{n_\alpha,m_\beta}^{(l+1+\alpha+\beta)}-u_{n_\alpha+1,m_\beta}^{(l+\beta)}}
 =\frac{a_2+u_{n_\alpha+1,m_\beta}^{(l+1+\beta)}-u_{n_\alpha+1,m_\beta+1}^{(l)}}{a_2+u_{n_\alpha,m_\beta}^{(l+1+\alpha+\beta)}-u_{n_\alpha,m_\beta+1}^{(l+\alpha)}}
\end{align}
for the additive potential $u$ as well as modified $\mathbb{Z}_\cN$ graded discrete integrable models
\begin{align}\label{dFX:v}
 a_1\left(\frac{v_{n_\alpha+1,m_\beta+1}^{(l)}}{v_{n_\alpha,m_\beta+1}^{(l+\alpha)}}-\frac{v_{n_\alpha+1,m_\beta}^{(l+1+\beta)}}{v_{n_\alpha,m_\beta}^{(l+1+\alpha+\beta)}}\right)
 =a_2\left(\frac{v_{n_\alpha+1,m_\beta+1}^{(l)}}{v_{n_\alpha+1,m_\beta}^{(l+\beta)}}-\frac{v_{n_\alpha,m_\beta+1}^{(l+1+\alpha)}}{v_{n_\alpha,m_\beta}^{(l+1+\alpha+\beta)}}\right)
\end{align}
\ese
for the quotient potential $v$, in which $l=0,1,\cdots,\cN-1$,
and we have $u_{n_\alpha,m_\beta}^{(l+\cN)}=u_{n_\alpha,m_\beta}^{(l)}$ and $v_{n_\alpha,m_\beta}^{(l+\cN)}=v_{n_\alpha,m_\beta}^{(l)}$.
Equation \eqref{dFX:v} is exactly the same as the one given in \cite{FX17};
while equation \eqref{dFX:u} is a different parametrisation, which allows us to take continuum limit.

As a remark, we note that \eqref{dGD}, i.e. the \ac{GD} equations of rank $\cN$, is actually a special case of \eqref{dFX} by taking $\alpha=\beta=0$.
Additionally, we can also prove identities
\begin{align}\label{dFX:FirstIntegral}
 \prod_{l=0}^{\cN-1}\left(a_1+u_{n_\alpha,m_\beta}^{(l+1+\alpha)}-u_{n_\alpha+1,m_\beta}^{(l)}\right)=a_1^\cN, \quad
 \prod_{l=0}^{\cN-1}\left(a_2+u_{n_\alpha,m_\beta}^{(l+1+\beta)}-u_{n_\alpha,m_\beta+1}^{(l)}\right)=a_2^\cN \quad \hbox{and} \quad
 \prod_{l=0}^{\cN-1}v_{n_\alpha,m_\beta}^{(l)}=1,
\end{align}
which are the deformations of \eqref{dGD:FirstIntegral}.
These identities are referred to as the first integrals of \eqref{dFX:u} and \eqref{dFX:v}, respectively.

The wave function $\phi$ and the adjoint wave function $\psi$ can also be changed through \eqref{dFX:CoV}.
Thus, following \eqref{dGD:Lax}, we obtain the Lax pair for \eqref{dFX:u}
\bse\label{dFX:Lax}
\begin{align}
 &\phi_{n_\alpha+1,m_\beta}^{(l)}=\left(a_1+u_{n_\alpha,m_\beta}^{(l+1+\alpha)}-u_{n_\alpha+1,m_\beta}^{(l)}\right)\phi_{n_\alpha,m_\beta}^{(l+\alpha)}+\ld\phi_{n_\alpha,m_\beta}^{(l+1+\alpha)}, \label{dFX:Lax1} \\
 &\phi_{n_\alpha,m_\beta+1}^{(l)}=\left(a_2+u_{n_\alpha,m_\beta}^{(l+1+\beta)}-u_{n_\alpha,m_\beta+1}^{(l)}\right)\phi_{n_\alpha,m_\beta}^{(l+\beta)}+\ld\phi_{n_\alpha,m_\beta}^{(l+1+\beta)}, \label{dFX:Lax2}
\end{align}
\ese
for $l=0,1,\cdots,\cN-1$, in which $\phi_{n_\alpha,m_\beta}^{(l+\cN)}=\phi_{n_\alpha,m_\beta}^{(l)}$.
The Lax pairs for equations  \eqref{dFX:u} and  \eqref{dFX:v}  can be obtained by the following transforms:
\bse\label{dFX:MT}
\begin{align}
 &a_1+u_{n_\alpha,m_\beta}^{(l+1+\alpha)}-u_{n_\alpha+1,m_\beta}^{(l)}
 =a_1\frac{v_{n_\alpha+1,m_\beta}^{(l)}}{v_{n_\alpha,m_\beta}^{(l+\alpha)}}
 =a_1\frac{\tau_{n_\alpha+1,m_\beta}^{(l+1)}\tau_{n_\alpha,m_\beta}^{(l+\alpha)}}{\tau_{n_\alpha+1,m_\beta}^{(l)}\tau_{n_\alpha,m_\beta}^{(l+1+\alpha)}}, \\
 &a_2+u_{n_\alpha,m_\beta}^{(l+1+\beta)}-u_{n_\alpha,m_\beta+1}^{(l)}
 =a_2\frac{v_{n_\alpha,m_\beta+1}^{(l)}}{v_{n_\alpha,m_\beta}^{(l+\beta)}}
 =a_2\frac{\tau_{n_\alpha,m_\beta+1}^{(l+1)}\tau_{n_\alpha,m_\beta}^{(l+\beta)}}{\tau_{n_\alpha,m_\beta+1}^{(l)}\tau_{n_\alpha,m_\beta}^{(l+1+\beta)}}.
\end{align}
\ese
These difference transforms in \eqref{dFX:MT} act as the Miura maps between the equations in \eqref{dFX},
which is a result of imposing the change of variables \eqref{dFX:CoV} on \eqref{dGD:MT}.
The linear equations in \eqref{dFX:Lax} are the $\mathbb{Z}_\cN$ graded discrete Lax pairs discussed in \cite{FX17}.
Similarly, applying the change of variables \eqref{dFX:CoV} to \eqref{dGD:Adjoint}, we obtain the adjoint Lax pair for \eqref{dFX:u}
\bse\label{dFX:Adjoint}
\begin{align}
 &\psi_{n_\alpha-1,m_\beta}^{(l)}=\left(a_1+u_{n_\alpha-1,m_\beta}^{(l)}-u_{n_\alpha,m_\beta}^{(l-1-\alpha)}\right)\psi_{n_\alpha,m_\beta}^{(l-\alpha)}+\ld\psi_{n_\alpha,m_\beta}^{(l-1-\alpha)}, \label{dFX:Adjoint1} \\
 &\psi_{n_\alpha,m_\beta-1}^{(l)}=\left(a_2+u_{n_\alpha,m_\beta-1}^{(l)}-u_{n_\alpha,m_\beta}^{(l-1-\beta)}\right)\psi_{n_\alpha,m_\beta}^{(l-\beta)}+\ld\psi_{n_\alpha,m_\beta}^{(l-1-\beta)}, \label{dFX:Adjoint2}
\end{align}
\ese
for $l=0,1,\cdots,\cN-1$ and the adjoint Lax pairs for \eqref{dFX:v} and \eqref{dFX:tau1} are derived with the help of the Miura maps in \eqref{dFX:MT}.
The adjoint Lax pair \eqref{dFX:Adjoint} was, however, not discussed in \cite{FX17}.

\begin{remark}
Without loss of generality, in the nonlinear equations \eqref{dFX} and linear equations \eqref{dFX:Lax} and \eqref{dFX:Adjoint}
we require that the indices in the superscript $(\cdot)$ attached to each variable (i.e. $u$, $v$, $\tau$, $\phi$ and $\psi$) are fixed at $0,1,\cdots,\cN-1$ subject to the corresponding periodicity condition.
\end{remark}

The Lax pair \eqref{dFX:Lax} and its adjoint \eqref{dFX:Adjoint} can be written in a more compact form,
with the help of the periodic matrix as follows
\begin{align}\label{PeriodicMatrix}
 \bSg=
 \begin{pmatrix}
 0 & 1 & & & \\
 & 0 & 1 & \\
 & & \ddots & \ddots & \\
 & & & \ddots & 1 \\
 1 & & & & 0
 \end{pmatrix}_{\cN\times\cN}.
\end{align}
By introducing notations
\begin{align*}
 \Phi_{n_\alpha,m_\beta}=\left(\phi_{n_\alpha,m_\beta}^{(0)},\phi_{n_\alpha,m_\beta}^{(1)},\cdots,\phi_{n_\alpha,m_\beta}^{(\cN-1)}\right)^\rT, \quad
 \Psi_{n_\alpha,m_\beta}=\left(\psi_{n_\alpha,m_\beta}^{(0)},\psi_{n_\alpha,m_\beta}^{(1)},\cdots,\psi_{n_\alpha,m_\beta}^{(\cN-1)}\right)^\rT,
\end{align*}
as well as
\begin{align*}
 \bU_{n_\alpha,m_\beta}=\diag\left(u_{n_\alpha,m_\beta}^{(0)},u_{n_\alpha,m_\beta}^{(1)},\cdots,u_{n_\alpha,m_\beta}^{(\cN-1)}\right),
\end{align*}
we can reformulate the linear system associated with \eqref{dFX:u}, namely \eqref{dFX:Lax}, and its adjoint \eqref{dFX:Adjoint} as
\begin{align*}
 &\Phi_{n_\alpha+1,m_\beta}=\left(\left(a_1+\bSg^{1+\alpha}\bU_{n_\alpha,m_\beta}\bSg^{-1-\alpha}-\bU_{n_\alpha+1,m_\beta}\right)\bSg^{\alpha}+\lambda\bSg^{1+\alpha}\right)\Phi_{n_\alpha,m_\beta}, \\
 &\Phi_{n_\alpha,m_\beta+1}=\left(\left(a_2+\bSg^{1+\beta}\bU_{n_\alpha,m_\beta}\bSg^{-1-\beta}-\bU_{n_\alpha,m_\beta+1}\right)\bSg^{\beta}+\lambda\bSg^{1+\beta}\right)\Phi_{n_\alpha,m_\beta},
\end{align*}
and
\begin{align*}
 &\Psi_{n_\alpha-1,m_\beta}=\left(\left(a_1+\bU_{n_\alpha-1,m_\beta}-\bSg^{-1-\alpha}\bU_{n_\alpha,m_\beta}\bSg^{1+\alpha}\right)\bSg^{-\alpha}+\lambda\bSg^{-1-\alpha}\right)\Psi_{n_\alpha,m_\beta}, \\
 &\Psi_{n_\alpha,m_\beta-1}=\left(\left(a_2+\bU_{n_\alpha,m_\beta-1}-\bSg^{-1-\beta}\bU_{n_\alpha,m_\beta}\bSg^{1+\beta}\right)\bSg^{-\beta}+\lambda\bSg^{-1-\beta}\right)\Psi_{n_\alpha,m_\beta},
\end{align*}
respectively, in which $a_1$ and $a_2$ denote $a_1\bI$ and $a_2\bI$, respectively.
If we adopt the notations
\begin{align*}
 \bV_{n_\alpha,m_\beta}=\left(v_{n_\alpha,m_\beta}^{(0)},v_{n_\alpha,m_\beta}^{(1)},\cdots,v_{n_\alpha,m_\beta}^{(\cN-1)}\right)
 \quad \hbox{and} \quad
 \bT_{n_\alpha,m_\beta}=\left(\tau_{n_\alpha,m_\beta}^{(0)},\tau_{n_\alpha,m_\beta}^{(1)},\cdots,\tau_{n_\alpha,m_\beta}^{(\cN-1)}\right),
\end{align*}
the Miura maps in \eqref{dFX:MT} can be written as their matrix form, namely
\begin{align*}
 &a_1+\bSg^{1+\alpha}\bU_{n_\alpha,m_\beta}\bSg^{-1-\alpha}-\bU_{n_\alpha+1,m_\beta} \\
 &\qquad=a_1\bV_{n_\alpha+1,m_\beta}\bSg^{\alpha}\bV_{n_\alpha,m_\beta}^{-1}\bSg^{-\alpha}
 =a_1\bSg\bT_{n_\alpha+1,m_\beta}\bSg^{-1}\bT_{n_\alpha+1,m_\beta}^{-1}
 \bSg^{\alpha}\bT_{n_\alpha,m_\beta}\bSg\bT_{n_\alpha,m_\beta}^{-1}\bSg^{-1-\alpha}, \\
 &a_2+\bSg^{1+\beta}\bU_{n_\alpha,m_\beta}\bSg^{-1-\beta}-\bU_{n_\alpha,m_\beta+1} \\
 &\qquad=a_2\bV_{n_\alpha,m_\beta+1}\bSg^{\beta}\bV_{n_\alpha,m_\beta}^{-1}\bSg^{-\beta}
 =a_2\bSg\bT_{n_\alpha,m_\beta+1}\bSg^{-1}\bT_{n_\alpha,m_\beta+1}^{-1}
 \bSg^{\beta}\bT_{n_\alpha,m_\beta}\bSg\bT_{n_\alpha,m_\beta}^{-1}\bSg^{-1-\beta},
\end{align*}
which in turn yield the matrix form of the Lax pairs and the adjoint ones for the modified equation \eqref{dFX:v} and the bilinear equation \eqref{dFX:tau1}.

The identities given in \eqref{dFX:FirstIntegral} restrict the degree of freedom of equations \eqref{dFX:u} and \eqref{dFX:v}.
Making use of the identity for $v$, one can always eliminate one component (for instance $v_{n_\alpha,m_\beta}^{(\cN-1)}$) in the modified system \eqref{dFX:v}.
However, it seems that we can only reduce the number of components in the unmodified system \eqref{dFX:u} when $(\alpha,\beta)=(0,0)$ or $(\alpha,\beta)=(\cN-1,\cN-1)$; in other cases, the system cannot be decoupled, and thus, the $u$-identities play the roles of additional constraints which the potential variable in \eqref{dFX:u} must obey. Below we list the two decoupled (i.e. $\cN-1$-component) systems of the additive potential $u$, in which the component $u_{n_\alpha,m_\beta}^{(\cN-1)}$ is eliminated.
\begin{enumerate}
\item{$(\alpha,\beta)=(0,0)$}
\begin{align*}
 &\frac{a_1+u_{n_0,m_0+1}^{(l+1)}-u_{n_0+1,m_0+1}^{(l)}}{a_1+u_{n_0,m_0}^{(l+1)}-u_{n_0+1,m_0}^{(l)}}
 =\frac{a_2+u_{n_0+1,m_0}^{(l+1)}-u_{n_0+1,m_0+1}^{(l)}}{a_2+u_{n_0,m_0}^{(l+1)}-u_{n_0,m_0+1}^{(l)}}, \quad
 l=0,1,\cdots,\cN-3, \\
 &\left(a_1+a_2+u_{n_0,m_0}^{(0)}-u_{n_0+1,m_0+1}^{(\cN-2)}\right)
 \left(a_1-a_2+u_{n_0,m_0+1}^{(\cN-2)}-u_{n_0+1,m_0}^{(\cN-2)}\right) \\
 &\qquad=\frac{a_1^\cN}{\prod_{l=0}^{\cN-3}\left(a_1+u_{n_0,m_0}^{(l+1)}-u_{n_0+1,m_0}^{(l)}\right)}
        -\frac{a_2^\cN}{\prod_{l=0}^{\cN-3}\left(a_2+u_{n_0,m_0}^{(l+1)}-u_{n_0,m_0+1}^{(l)}\right)};
\end{align*}
\item{$(\alpha,\beta)=(\cN-1,\cN-1)$}
\begin{align*}
 &\frac{a_1+u_{n_0,m_0+1}^{(0)}-u_{n_0+1,m_0+1}^{(0)}}{a_2+u_{n_0+1,m_0}^{(0)}-u_{n_0+1,m_0+1}^{(0)}}
 =\frac{a_1^\cN\prod_{l=0}^{\cN-2}\left(a_2+u_{n_0,m_0}^{(l)}-u_{n_0,m_0+1}^{(l)}\right)}
       {a_2^\cN\prod_{l=0}^{\cN-2}\left(a_1+u_{n_0,m_0}^{(l)}-u_{n_0+1,m_0}^{(l)}\right)}, \\
 &\frac{a_1+u_{n_0,m_0+1}^{(l)}-u_{n_0+1,m_0+1}^{(l)}}{a_2+u_{n_0+1,m_0}^{(l)}-u_{n_0+1,m_0+1}^{(l)}}
 =\frac{a_1+u_{n_0,m_0}^{(l-1)}-u_{n_0+1,m_0}^{(l-1)}}{a_2+u_{n_0,m_0}^{(l-1)}-u_{n_0,m_0+1}^{(l-1)}}, \quad
 l=1,2,\cdots,\cN-3.
\end{align*}
\end{enumerate}

We list some concrete examples for $\cN=2$ and $\cN=3$ within this framework in appendix \ref{S:Exampl} for reference.
Among these examples, the classes of $(\alpha,\beta)=(0,0)$, $(\alpha,\beta)=(0,\cN-1)$ and $(\alpha,\beta)=(\cN-1,\cN-1)$ amount to the discrete \ac{GD} hierarchy \cite{NPCQ92}, the discrete-time \ac{2DTL} equations \cite{Fu18a} and the discrete Schwarzian \ac{GD} hierarchy \cite{ALN12}, respectively; while the other classes of equations are the `new' integrable discrete equations which so far only appeared in \cite{FX17} and were not considered by others.
Our parametrisations/expressions of these equations could sometimes be slightly different from those given by Fordy and Xenitidis,
for convenience of constructing their \ac{DT}s  in the forthcoming sections, i.e. sections \ref{S:DT} .
In addition, by our way, we also provide the equations based on the bilinear potential as well as their Lax formalism which were not given in \cite{FX17}.

\subsection{Associated continuous equations}

By performing the same reduction, namely taking $a_3=0$, imposing $\tau_{n, m, h+\cN}=\tau_{n, m, h}$, $\phi_{n, m, h+\cN}=\lambda^{\cN}\phi_{n, m, h}$, $\psi_{n, m, h+\cN}=\lambda^{-\cN}\psi_{n, m, h}$ and introducing the change of variables \eqref{dFX:CoV}, we can from \eqref{sdKP:Lax} and \eqref{sdKP:Adjoint} derive the continuous linear equation
\begin{align}\label{cFX:Lax}
 \partial_x\phi_{n_\alpha,m_\beta}^{(l)}
 =\left(u_{n_\alpha,m_\beta}^{(l+1)}-u_{n_\alpha,m_\beta}^{(l)}\right)\phi_{n_\alpha,m_\beta}^{(l)}
 +\lambda\phi_{n_\alpha,m_\beta}^{(l+1)}
\end{align}
and its adjoint form
\begin{align}\label{cFX:Adjoint}
 \partial_x\psi_{n_\alpha,m_\beta}^{(l)}
 =\left(u_{n_\alpha,m_\beta}^{(l-1)}-u_{n_\alpha,m_\beta}^{(l)}\right)\psi_{n_\alpha,m_\beta}^{(l)}
 -\lambda\psi_{n_\alpha,m_\beta}^{(l-1)},
\end{align}
respectively, where $l=0,1,\cdots,\cN-1$, and all the variables satisfy their respective periodicity condition.
For consistency with the $\mathbb{Z}_\cN$ graded discrete Lax pairs \eqref{dFX:Lax} and \eqref{dFX:Adjoint},
we require that all the components are still fixed at $0,1,\cdots,\cN-1$.
Following the idea of deriving \eqref{sdKP:MT}, in this case we have the bilinear transforms
\begin{align}\label{cFX:BLT}
 u_{n_\alpha,m_\beta}^{(l)}=\partial_x\left(\ln\tau_{n_\alpha,m_\beta}^{(l)}\right) \quad \hbox{and} \quad
 v_{n_\alpha,m_\beta}^{(l)}=\frac{\tau_{n_\alpha,m_\beta}^{(l+1)}}{\tau_{n_\alpha,m_\beta}^{(l)}}
\end{align}
for $l=0,1,\cdots,\cN-1$, which result in the Miura maps between the additive potential $u$, quotient potential $v$
and the bilinear potential $\tau$ as follows:
\begin{align}\label{sdFX:MT}
 u_{n_\alpha,m_\beta}^{(l+1)}-u_{n_\alpha,m_\beta}^{(l)}=\partial_x\left(\ln v_{n_\alpha,m_\beta}^{(l)}\right)
 =\partial_x\left(\ln\tau_{n_\alpha,m_\beta}^{(l+1)}\right)-\partial_x\left(\ln\tau_{n_\alpha,m_\beta}^{(l)}\right).
\end{align}
Such difference relations can be used to replace the potential $u$ in \eqref{cFX:Lax} and \eqref{cFX:Adjoint}
by the other two potentials, leading to continuous linear equation in their respective semi-discrete Lax pair.
Equations \eqref{cFX:Lax} and \eqref{cFX:Adjoint} can also be written in matrix form
with the help of the periodic matrix $\bSg$ as follows:
\begin{align*}
 &\partial_x\Phi_{n_\alpha,m_\beta}=\left(\bSg\bU_{n_\alpha,m_\beta}\bSg^{-1}-\bU_{n_\alpha,m_\beta}+\lambda\bSg\right)\Phi_{n_\alpha,m_\beta},
 &\partial_x\Psi_{n_\alpha,m_\beta}=\left(\bSg^{-1}\bU_{n_\alpha,m_\beta}\bSg-\bU_{n_\alpha,m_\beta}+\lambda\bSg^{-1}\right)\Psi_{n_\alpha,m_\beta}.
\end{align*}
Similarly, the Miura maps \eqref{sdFX:MT} can be written as
\begin{align*}
 \bSg\bU_{n_\alpha,m_\beta}\bSg^{-1}-\bU_{n_\alpha,m_\beta}=\partial_x\ln\bV_{n_\alpha,m_\beta}
 =\partial_x\ln\left(\bSg\bT_{n_\alpha,m_\beta}\bSg^{-1}\right)-\partial_x\ln\bT_{n_\alpha,m_\beta},
\end{align*}
leading the spectral problems of $\bU$ to the $\bV$- and $\bT$-counterparts.
Such a matrix structure actually coincides with the Lax structure of the \ac{2DTL} equations, see \cite{FG80};
in other words, the continuous $\mathbb{Z}_\cN$ graded Lax matrices deep down describe the structure of
the affine Lie algebra $A_\cN^{(1)}$, as can be seen that these Lax matrices have trace zero due to the periodicity conditions.

The compatibility condition between \eqref{cFX:Lax} and \eqref{dFX:Lax1} gives rise to
semi-discrete equations of $x$ and $n_\alpha$ as follows:
\bse\label{sdFX-1}
\begin{align}
 &\partial_x\left(u_{n_\alpha,m_\beta}^{(l+1+\alpha)}-u_{n_\alpha+1,m_\beta}^{(l)}\right)
 =\left(a_1+u_{n_\alpha,m_\beta}^{(l+1+\alpha)}-u_{n_\alpha+1,m_\beta}^{(l)}\right)
 \left(u_{n_\alpha,m_\beta}^{(l+\alpha)}-u_{n_\alpha,m_\beta}^{(l+1+\alpha)}-u_{n_\alpha+1,m_\beta}^{(l)}+u_{n_\alpha+1,m_\beta}^{(l+1)}\right), \label{sdFX:u} \\
 &\partial_x\left(\ln v_{n_\alpha,m_\beta}^{(l+1+\alpha)}-\ln v_{n_\alpha+1,m_\beta}^{(l)}\right)
 =a_1\left(\frac{v_{n_\alpha+1,m_\beta}^{(l+1)}}{v_{n_\alpha,m_\beta}^{(l+1+\alpha)}}-\frac{v_{n_\alpha+1,m_\beta}^{(l)}}{v_{n_\alpha,m_\beta}^{(l+\alpha)}}\right), \label{sdFX:v} \\
 &\tau_{n_\alpha+1,m_\beta}^{(l)}\partial_x\tau_{n_\alpha,m_\beta}^{(l+1+\alpha)}-\tau_{n_\alpha,m_\beta}^{(l+1+\alpha)}\partial_x\tau_{n_\alpha+1,m_\beta}^{(l)}
 =a_1\left(\tau_{n_\alpha,m_\beta}^{(l+\alpha)}\tau_{n_\alpha+1,m_\beta}^{(l+1)}-\tau_{n_\alpha,m_\beta}^{(l+1+\alpha)}\tau_{n_\alpha+1,m_\beta}^{(l)}\right). \label{sdFX:tau}
\end{align}
\ese
for $l=0,1,\cdots,\cN-1$.

The compatibility condition between \eqref{cFX:Lax} and \eqref{dFX:Lax2} gives rise to
semi-discrete equations of $x$ and $m_\beta$ as follows:
\bse\label{sdFX-2}
\begin{align}
 &\partial_x\left(u_{n_\alpha,m_\beta}^{(l+1+\beta)}-u_{n_\alpha,m_\beta+1}^{(l)}\right)
 =\left(a_2+u_{n_\alpha,m_\beta}^{(l+1+\beta)}-u_{n_\alpha,m_\beta+1}^{(l)}\right)
 \left(u_{n_\alpha,m_\beta}^{(l+\beta)}-u_{n_\alpha,m_\beta}^{(l+1+\beta)}-u_{n_\alpha,m_\beta+1}^{(l)}+u_{n_\alpha,m_\beta+1}^{(l+1)}\right), \label{sdFX:u} \\
 &\partial_x\left(\ln v_{n_\alpha,m_\beta}^{(l+1+\beta)}-\ln v_{n_\alpha,m_\beta+1}^{(l)}\right)
 =a_2\left(\frac{v_{n_\alpha,m_\beta+1}^{(l+1)}}{v_{n_\alpha,m_\beta}^{(l+1+\beta)}}-\frac{v_{n_\alpha,m_\beta+1}^{(l)}}{v_{n_\alpha,m_\beta}^{(l+\beta)}}\right), \label{sdFX:v} \\
 &\tau_{n_\alpha,m_\beta+1}^{(l)}\partial_x\tau_{n_\alpha,m_\beta}^{(l+1+\beta)}-\tau_{n_\alpha,m_\beta}^{(l+1+\beta)}\partial_x\tau_{n_\alpha,m_\beta+1}^{(l)}
 =a_2\left(\tau_{n_\alpha,m_\beta}^{(l+\beta)}\tau_{n_\alpha,m_\beta+1}^{(l+1)}-\tau_{n_\alpha,m_\beta}^{(l+1+\beta)}\tau_{n_\alpha,m_\beta+1}^{(l)}\right). \label{sdFX:tau}
\end{align}
\ese
for $l=0,1,\cdots,\cN-1$.

These equations also arise as the compatibility of \eqref{cFX:Adjoint} and \eqref{dFX:Adjoint1},  \eqref{cFX:Adjoint} and \eqref{dFX:Adjoint2}, respectively.
By direct calculation, one can verify that equations \eqref{sdFX-1} and \eqref{sdFX-2} are compatible with \eqref{dFX}, respectively.

\section{Darboux transformations}\label{S:DT}

\subsection{Basic Darboux transformations}

We now consider the Darboux transformations of the $\mathbb{Z}_\cN$ graded discrete integrable systems shown above in order to provide exact solutions.
The \ac{DT}s will be constructed on the level of tau function,
which then naturally gain those additive and quotient potentials (or say unmodified and modified potentials) through the transforms \eqref{cFX:BLT}.

As explained in section \ref{S:ZN} that with the help of the Miura maps, the $\mathbb{Z}_\cN$ graded discrete Lax pairs (including associated continuous part)  in tau function are given by
\bse\label{dFX:Lax-tau}
\begin{align}
 &\phi_{n_\alpha+1,m_\beta}^{(l)}=a_1\frac{\tau_{n_\alpha,m_\beta}^{(l+\alpha)}\tau_{n_\alpha+1,m_\beta}^{(l+1)}}
 {\tau_{n_\alpha,m_\beta}^{(l+\alpha+1)}\tau_{n_\alpha+1,m_\beta}^{(l)}}\phi_{n_\alpha,m_\beta}^{(l+\alpha)}+\ld\phi_{n_\alpha,m_\beta}^{(l+\alpha+1)}, \label{dFX:Lax1-tau} \\
 &\phi_{n_\alpha,m_\beta+1}^{(l)}=a_2\frac{\tau_{n_\alpha,m_\beta}^{(l+\beta)}\tau_{n_\alpha,m_\beta+1}^{(l+1)}}
 {\tau_{n_\alpha,m_\beta}^{(l+\beta+1)}\tau_{n_\alpha,m_\beta+1}^{(l)}}\phi_{n_\alpha,m_\beta}^{(l+\beta)}+\ld\phi_{n_\alpha,m_\beta}^{(l+\beta+1)}, \label{dFX:Lax2-tau}\\
 &\partial_x\phi_{n_\alpha,m_\beta}^{(l)}=\left(\partial_x\ln\tau_{n_\alpha,m_\beta}^{(l+1)}-\partial_x\ln\tau_{n_\alpha,m_\beta}^{(l)}\right)
 \phi_{n_\alpha,m_\beta}^{(l)}+\ld\phi_{n_\alpha,m_\beta}^{(l+1)},\label{dFX:Lax3-tau}
\end{align}
\ese
where $l=0,1,\cdots,\cN-1$.
\begin{proposition}\label{dFX:1DT}
Let the non-zero function $\theta_{n_\alpha,m_\beta}^{\left(l,1\right)}$ satisfy the linear equations \eqref{dFX:Lax-tau} for given $\tau_{n_\alpha,m_\beta}^{(l)}$ with $\ld=\ld_1$,  for $l=0, 1, \cdots, \cN-1$. Then the following \ac{DT}s
\begin{align}
\begin{split}\label{dFX:equivalence}
\phi_{n_\alpha,m_\beta}^{\left(l\right)}\rightarrow \wt{\phi}_{n_\alpha,m_\beta}^{\left(l\right)}=~&\ld\phi_{n_\alpha,m_\beta}^{\left(l+\alpha+\beta+1\right)}-\ld_1\theta_{n_\alpha,m_\beta}^{\left(l+\alpha+\beta+1,1\right)}
{\theta_{n_\alpha,m_\beta}^{\left(l+\alpha+\beta,1\right)}}^{-1}\phi_{n_\alpha,m_\beta}^{\left(l+\alpha+\beta\right)}\\
=~&\partial_x\phi_{n_\alpha,m_\beta}^{(l+\alpha+\beta)}-\partial_x\theta_{n_\alpha,m_\beta}^{(l+\alpha+\beta,1)}
{\theta_{n_\alpha,m_\beta}^{\left(l+\alpha+\beta,1\right)}}^{-1}\phi_{n_\alpha,m_\beta}^{\left(l+\alpha+\beta\right)},
\end{split}\\
\tau_{n_\alpha,m_\beta}^{(l)}\rightarrow\wt{\tau}_{n_\alpha,m_\beta}^{(l)}=~&\theta_{n_\alpha,m_\beta}^{\left(l+\alpha+\beta,1\right)}\tau_{n_\alpha,m_\beta}^{(l+\alpha+\beta)},\label{dFX:tau}
\end{align}
leave \eqref{dFX:Lax-tau} invariant, where the component counter $l=0, 1, \cdots, \cN-1$, and the parameters  $\alpha,\beta\in\{0,1,\cdots,\cN-1\}$.
\end{proposition}

\begin{proof}
The proof of the \ac{DT}s is a straightforward computation. Firstly, we prove that the two expressions (discrete and continuous expression) for $\wt\phi_{n_\alpha,m_\beta}^{(l)}$ in the \ac{DT} \eqref{dFX:equivalence} are equal.
From the continuous part in linear system \eqref{dFX:Lax3-tau} when $l \rightarrow l+\alpha+\beta$, we have
\begin{align*}
\partial_x\phi_{n_\alpha,m_\beta}^{(l+\alpha+\beta)}{\phi_{n_\alpha,m_\beta}^{(l+\alpha+\beta)}}^{-1}
=\left(\frac{\partial_x\tau_{n_\alpha,m_\beta}^{(l+\alpha+\beta+1)}}{\tau_{n_\alpha,m_\beta}^{(l+\alpha+\beta+1)}}
-\frac{\partial_x\tau_{n_\alpha,m_\beta}^{(l+\alpha+\beta)}}{\tau_{n_\alpha,m_\beta}^{(l+\alpha+\beta)}}\right)
+\lambda\phi_{n_\alpha,m_\beta}^{(l+\alpha+\beta+1)}{\phi_{n_\alpha,m_\beta}^{(l+\alpha+\beta)}}^{-1}.
\end{align*}
Similarly, when $\ld=\ld_1$, we have
\begin{align*}
\partial_x{\theta}_{n_\alpha,m_\beta}^{(l+\alpha+\beta,1)}{{\theta}_{n_\alpha,m_\beta}^{(l+\alpha+\beta,1)}}^{-1}
=\left(\frac{\partial_x\tau_{n_\alpha,m_\beta}^{(l+\alpha+\beta+1)}}{\tau_{n_\alpha,m_\beta}^{(l+\alpha+\beta+1)}}
-\frac{\partial_x\tau_{n_\alpha,m_\beta}^{(l+\alpha+\beta)}}{\tau_{n_\alpha,m_\beta}^{(l+\alpha+\beta)}}\right)
+\ld_1{\theta}_{n_\alpha,m_\beta}^{(l+\alpha+\beta+1,1)}{{\theta}_{n_\alpha,m_\beta}^{(l+\alpha+\beta,1)}}^{-1}.
\end{align*}
Eliminating the $\tau$ function in the above two equations, then we get the two expressions for $\wt\phi_{n_\alpha,m_\beta}^{(l)}$ in the \ac{DT} \eqref{dFX:equivalence}.

Substituting the discrete expression (the first expression in \eqref{dFX:equivalence}) for $\wt\phi_{n_\alpha,m_\beta}^{(l)}$ into \eqref{dFX:Lax1-tau} we get
\begin{align}\label{dFX:1-DT-1}
\begin{split}
&\ld\left(\phi_{n_\alpha+1,m_\beta}^{\left(l+\alpha+\beta+1\right)}-\ld\phi_{n_\alpha,m_\beta}^{\left(l+2\alpha+\beta+2\right)}+\ld_1\theta_{n_\alpha,m_\beta}^{\left(l+2\alpha+\beta+2,1\right)}
{\theta_{n_\alpha,m_\beta}^{\left(l+2\alpha+\beta+1,1\right)}}^{-1}\phi_{n_\alpha,m_\beta}^{\left(l+2\alpha+\beta+1\right)}\right)-\ld_1\theta_{n_\alpha+1,m_\beta}^{\left(l+\alpha+\beta+1,1\right)}
{\theta_{n_\alpha+1,m_\beta}^{\left(l+\alpha+\beta,1\right)}}^{-1}\phi_{n_\alpha+1,m_\beta}^{\left(l+\alpha+\beta\right)}\\
&=a_1\frac{\wt{\tau}_{n_\alpha,m_\beta}^{(l+\alpha)}\wt{\tau}_{n_\alpha+1,m_\beta}^{(l+1)}}
 {\wt{\tau}_{n_\alpha,m_\beta}^{(l+\alpha+1)}\wt{\tau}_{n_\alpha+1,m_\beta}^{(l)}}\left(\ld\phi_{n_\alpha,m_\beta}^{\left(l+\alpha+\beta+1\right)}-\ld_1\theta_{n_\alpha,m_\beta}^{\left(l+\alpha+\beta+1,1\right)}
{\theta_{n_\alpha,m_\beta}^{\left(l+\alpha+\beta,1\right)}}^{-1}\phi_{n_\alpha,m_\beta}^{\left(l+\alpha+\beta\right)}\right).
\end{split}
\end{align}
Taking into account the linear equation \eqref{dFX:Lax1-tau} when $l \rightarrow l+\alpha+\beta+1$ and $l \rightarrow l+\alpha+\beta$ respectively
\begin{align*}
\phi_{n_\alpha+1,m_\beta}^{(l+\alpha+\beta+1)}&=a_1\frac{\tau_{n_\alpha,m_\beta}^{(l+2\alpha+\beta+1)}\tau_{n_\alpha+1,m_\beta}^{(l+\alpha+\beta+2)}}
 {\tau_{n_\alpha,m_\beta}^{(l+2\alpha+\beta+2)}\tau_{n_\alpha+1,m_\beta}^{(l+\alpha+\beta+1)}}\phi_{n_\alpha,m_\beta}^{(l+2\alpha+\beta+1)}+\ld\phi_{n_\alpha,m_\beta}^{(l+2\alpha+\beta+2)}, \\
 \phi_{n_\alpha+1,m_\beta}^{(l+\alpha+\beta)}&=a_1\frac{\tau_{n_\alpha,m_\beta}^{(l+2\alpha+\beta)}\tau_{n_\alpha+1,m_\beta}^{(l+\alpha+\beta+1)}}
 {\tau_{n_\alpha,m_\beta}^{(l+2\alpha+\beta+1)}\tau_{n_\alpha+1,m_\beta}^{(l+\alpha+\beta)}}\phi_{n_\alpha,m_\beta}^{(l+2\alpha+\beta)}+\ld\phi_{n_\alpha,m_\beta}^{(l+2\alpha+\beta+1)},
\end{align*}
we can rewrite \eqref{dFX:1-DT-1} in the form
\begin{align}\label{dFX:1-DT-2}
\begin{split}
&\ld\left(a_1\frac{\tau_{n_\alpha,m_\beta}^{(l+2\alpha+\beta+1)}\tau_{n_\alpha+1,m_\beta}^{(l+\alpha+\beta+2)}}
 {\tau_{n_\alpha,m_\beta}^{(l+2\alpha+\beta+2)}\tau_{n_\alpha+1,m_\beta}^{(l+\alpha+\beta+1)}}\phi_{n_\alpha,m_\beta}^{(l+2\alpha+\beta+1)}+\ld_1\theta_{n_\alpha,m_\beta}^{\left(l+2\alpha+\beta+2,1\right)}
{\theta_{n_\alpha,m_\beta}^{\left(l+2\alpha+\beta+1,1\right)}}^{-1}\phi_{n_\alpha,m_\beta}^{\left(l+2\alpha+\beta+1\right)}\right)\\
&-\ld_1\theta_{n_\alpha+1,m_\beta}^{\left(l+\alpha+\beta+1,1\right)}{\theta_{n_\alpha+1,m_\beta}^{\left(l+\alpha+\beta,1\right)}}^{-1}
\left(a_1\frac{\tau_{n_\alpha,m_\beta}^{(l+2\alpha+\beta)}\tau_{n_\alpha+1,m_\beta}^{(l+\alpha+\beta+1)}}
 {\tau_{n_\alpha,m_\beta}^{(l+2\alpha+\beta+1)}\tau_{n_\alpha+1,m_\beta}^{(l+\alpha+\beta)}}\phi_{n_\alpha,m_\beta}^{(l+2\alpha+\beta)}+\ld\phi_{n_\alpha,m_\beta}^{(l+2\alpha+\beta+1)}\right)\\
&=a_1\frac{\wt{\tau}_{n_\alpha,m_\beta}^{(l+\alpha)}\wt{\tau}_{n_\alpha+1,m_\beta}^{(l+1)}}
 {\wt{\tau}_{n_\alpha,m_\beta}^{(l+\alpha+1)}\wt{\tau}_{n_\alpha+1,m_\beta}^{(l)}}\left(\ld\phi_{n_\alpha,m_\beta}^{\left(l+2\alpha+\beta+1\right)}-\ld_1\theta_{n_\alpha,m_\beta}^{\left(l+2\alpha+\beta+1,1\right)}
{\theta_{n_\alpha,m_\beta}^{\left(l+2\alpha+\beta,1\right)}}^{-1}\phi_{n_\alpha,m_\beta}^{\left(l+2\alpha+\beta\right)}\right).
\end{split}
\end{align}
Taking into account the linear equation \eqref{dFX:Lax1-tau} with $\ld=\ld_1$ when $l \rightarrow l+\alpha+\beta+1$ and $l \rightarrow l+\alpha+\beta$ respectively
\begin{align*}
\theta_{n_\alpha+1,m_\beta}^{(l+\alpha+\beta+1,1)}&=a_1\frac{\tau_{n_\alpha,m_\beta}^{(l+2\alpha+\beta+1)}\tau_{n_\alpha+1,m_\beta}^{(l+\alpha+\beta+2)}}
 {\tau_{n_\alpha,m_\beta}^{(l+2\alpha+\beta+2)}\tau_{n_\alpha+1,m_\beta}^{(l+\alpha+\beta+1)}}\theta_{n_\alpha,m_\beta}^{(l+2\alpha+\beta+1,1)}+\ld_1\theta_{n_\alpha,m_\beta}^{(l+2\alpha+\beta+2,1)}, \\
 \theta_{n_\alpha+1,m_\beta}^{(l+\alpha+\beta,1)}&=a_1\frac{\tau_{n_\alpha,m_\beta}^{(l+2\alpha+\beta)}\tau_{n_\alpha+1,m_\beta}^{(l+\alpha+\beta+1)}}
 {\tau_{n_\alpha,m_\beta}^{(l+2\alpha+\beta+1)}\tau_{n_\alpha+1,m_\beta}^{(l+\alpha+\beta)}}\theta_{n_\alpha,m_\beta}^{(l+2\alpha+\beta,1)}+\ld_1\theta_{n_\alpha,m_\beta}^{(l+2\alpha+\beta+1,1)},
\end{align*}
which can be equivalently written as
\begin{align*}
&a_1\frac{\tau_{n_\alpha,m_\beta}^{(l+2\alpha+\beta+1)}\tau_{n_\alpha+1,m_\beta}^{(l+\alpha+\beta+2)}}
 {\tau_{n_\alpha,m_\beta}^{(l+2\alpha+\beta+2)}\tau_{n_\alpha+1,m_\beta}^{(l+\alpha+\beta+1)}}\phi_{n_\alpha,m_\beta}^{(l+2\alpha+\beta+1)}
 +\ld_1\theta_{n_\alpha,m_\beta}^{(l+2\alpha+\beta+2,1)}{\theta_{n_\alpha,m_\beta}^{(l+2\alpha+\beta+1,1)}}^{-1}\phi_{n_\alpha,m_\beta}^{(l+2\alpha+\beta+1)}
\\
&=\theta_{n_\alpha+1,m_\beta}^{(l+\alpha+\beta+1,1)}{\theta_{n_\alpha,m_\beta}^{(l+2\alpha+\beta+1,1)}}^{-1}\phi_{n_\alpha,m_\beta}^{(l+2\alpha+\beta+1)}, \\
&a_1\frac{\theta_{n_\alpha,m_\beta}^{(l+2\alpha+\beta,1)}\tau_{n_\alpha,m_\beta}^{(l+2\alpha+\beta)}\theta_{n_\alpha+1,m_\beta}^{(l+\alpha+\beta+1,1)}\tau_{n_\alpha+1,m_\beta}^{(l+\alpha+\beta+1)}}
 {{\theta_{n_\alpha,m_\beta}^{(l+2\alpha+\beta+1,1)}}\tau_{n_\alpha,m_\beta}^{(l+2\alpha+\beta+1)}{\theta_{n_\alpha+1,m_\beta}^{(l+\alpha+\beta,1)}}\tau_{n_\alpha+1,m_\beta}^{(l+\alpha+\beta)}}
\phi_{n_\alpha,m_\beta}^{(l+2\alpha+\beta+1)}
=\frac{\theta_{n_\alpha+1,m_\beta}^{(l+\alpha+\beta+1,1)}}
 {{\theta_{n_\alpha,m_\beta}^{(l+2\alpha+\beta+1,1)}}}
\phi_{n_\alpha,m_\beta}^{(l+2\alpha+\beta+1)}
-\ld_1\frac{\theta_{n_\alpha+1,m_\beta}^{(l+\alpha+\beta+1,1)}}
 {{\theta_{n_\alpha+1,m_\beta}^{(l+\alpha+\beta,1)}}}
\phi_{n_\alpha,m_\beta}^{(l+2\alpha+\beta+1)},
\end{align*}
we can represent \eqref{dFX:1-DT-2} in the form
\begin{align*}
\begin{split}
&a_1\frac{\theta_{n_\alpha,m_\beta}^{\left(l+2\alpha+\beta,1\right)}\tau_{n_\alpha,m_\beta}^{(l+2\alpha+\beta)}
\theta_{n_\alpha+1,m_\beta}^{\left(l+\alpha+\beta+1,1\right)}\tau_{n_\alpha+1,m_\beta}^{(l+\alpha+\beta+1)}}
{\theta_{n_\alpha,m_\beta}^{\left(l+2\alpha+\beta+1,1\right)}\tau_{n_\alpha,m_\beta}^{(l+2\alpha+\beta+1)}
\theta_{n_\alpha+1,m_\beta}^{\left(l+\alpha+\beta,1\right)}\tau_{n_\alpha+1,m_\beta}^{(l+\alpha+\beta)}}
\left(\ld\phi_{n_\alpha,m_\beta}^{\left(l+2\alpha+\beta+1\right)}-\ld_1\theta_{n_\alpha,m_\beta}^{\left(l+2\alpha+\beta+1,1\right)}
{\theta_{n_\alpha,m_\beta}^{\left(l+2\alpha+\beta,1\right)}}^{-1}\phi_{n_\alpha,m_\beta}^{\left(l+2\alpha+\beta\right)}\right)\\
&=a_1\frac{\wt{\tau}_{n_\alpha,m_\beta}^{(l+\alpha)}\wt{\tau}_{n_\alpha+1,m_\beta}^{(l+1)}}
 {\wt{\tau}_{n_\alpha,m_\beta}^{(l+\alpha+1)}\wt{\tau}_{n_\alpha+1,m_\beta}^{(l)}}\left(\ld\phi_{n_\alpha,m_\beta}^{\left(l+2\alpha+\beta+1\right)}-\ld_1\theta_{n_\alpha,m_\beta}^{\left(l+2\alpha+\beta+1,1\right)}
{\theta_{n_\alpha,m_\beta}^{\left(l+2\alpha+\beta,1\right)}}^{-1}\phi_{n_\alpha,m_\beta}^{\left(l+2\alpha+\beta\right)}\right),
\end{split}
\end{align*}
which is evidently satisfied by virtue of the discrete \ac{DT} of $\tau$ in \eqref{dFX:tau}.

Similarly, the $\wt\phi_{n_\alpha,m_\beta}^{(l)}$ satisfies \eqref{dFX:Lax2-tau} with the discrete \ac{DT}s given by \eqref{dFX:equivalence} and \eqref{dFX:tau}.

Next, we prove that $\wt\phi_{n_\alpha,m_\beta}^{(l)}$ satisfies \eqref{dFX:Lax3-tau} with the \ac{DT}s given by \eqref{dFX:equivalence} and \eqref{dFX:tau}.
Substituting the discrete and continuous expression in \eqref{dFX:equivalence}) for $\wt\phi_{n_\alpha,m_\beta}^{(l)}$  simultaneously into \eqref{dFX:Lax3-tau} we get
\begin{align*}
\begin{split}
&\partial_x \left(\ld\phi_{n_\alpha,m_\beta}^{\left(l+\alpha+\beta+1\right)}\!-\!\ld_1\theta_{n_\alpha,m_\beta}^{\left(l+\alpha+\beta+1,1\right)}
{\theta_{n_\alpha,m_\beta}^{\left(l+\alpha+\beta,1\right)}}^{-1}\phi_{n_\alpha,m_\beta}^{\left(l+\alpha+\beta\right)}\right)
\!-\!\ld\left(\partial_x\phi_{n_\alpha,m_\beta}^{(l+\alpha+\beta+1)}-\partial_x\theta_{n_\alpha,m_\beta}^{(l+\alpha+\beta+1,1)}
{\theta_{n_\alpha,m_\beta}^{\left(l+\alpha+\beta+1,1\right)}}^{-1}\phi_{n_\alpha,m_\beta}^{\left(l+\alpha+\beta+1\right)}\right)\\
&=\left(\partial_x\ln\wt{\tau}_{n_\alpha,m_\beta}^{(l+1)}-\partial_x\ln\wt{\tau}_{n_\alpha,m_\beta}^{(l)})\right)
 \wt{\phi}_{n_\alpha,m_\beta}^{(l)}.
\end{split}
\end{align*}
After obvious simplifications we may represent the above equation in the form
\begin{align}\label{dFX:1-DT-3}
\left(\frac{\partial_x\theta_{n_\alpha,m_\beta}^{\left(l+\alpha+\beta+1,1\right)}}{\theta_{n_\alpha,m_\beta}^{\left(l+\alpha+\beta+1,1\right)}}
-\ld_1\frac{\theta_{n_\alpha,m_\beta}^{\left(l+\alpha+\beta+1,1\right)}}{\theta_{n_\alpha,m_\beta}^{\left(l+\alpha+\beta,1\right)}}\right)\wt{\phi}_{n_\alpha,m_\beta}^{(l)}
=\left(\partial_x\ln\wt{\tau}_{n_\alpha,m_\beta}^{(l+1)}-\partial_x\ln\wt{\tau}_{n_\alpha,m_\beta}^{(l)})\right)
\wt{\phi}_{n_\alpha,m_\beta}^{(l)}.
 \end{align}
Taking into account the linear equation \eqref{dFX:Lax3-tau} with $\ld=\ld_1$ when $l \rightarrow l+\alpha+\beta$
\begin{align*}
\partial_x\theta_{n_\alpha,m_\beta}^{(l+\alpha+\beta,1)}=\left(\partial_x\ln\tau_{n_\alpha,m_\beta}^{(l+\alpha+\beta+1,1)}-\partial_x\ln\tau_{n_\alpha,m_\beta}^{(l+\alpha+\beta)}\right)
 \theta_{n_\alpha,m_\beta}^{(l+\alpha+\beta,0)}+\ld_1\theta_{n_\alpha,m_\beta}^{(l+\alpha+\beta+1,1)},
\end{align*}
which can be equivalently written as
\begin{align*}
\frac{\partial_x\theta_{n_\alpha,m_\beta}^{(l+\alpha+\beta,1)}}{\theta_{n_\alpha,m_\beta}^{(l+\alpha+\beta,1)}}
-\ld_1\frac{\theta_{n_\alpha,m_\beta}^{(l+\alpha+\beta+1,1)}}{\theta_{n_\alpha,m_\beta}^{(l+\alpha+\beta,1)}}
=\left(\partial_x\ln\tau_{n_\alpha,m_\beta}^{(l+\alpha+\beta+1)}-\partial_x\ln\tau_{n_\alpha,m_\beta}^{(l+\alpha+\beta)}\right),
\end{align*}
we can rewrite \eqref{dFX:1-DT-3} in the form
\begin{align*}
\left(\partial_x\ln\left(\theta_{n_\alpha,m_\beta}^{(l+\alpha+\beta+1)}\tau_{n_\alpha,m_\beta}^{(l+\alpha+\beta+1)}\right)
-\partial_x\ln\left(\theta_{n_\alpha,m_\beta}^{(l+\alpha+\beta)}\tau_{n_\alpha,m_\beta}^{(l+\alpha+\beta)}\right)\right)\wt{\phi}_{n_\alpha,m_\beta}^{(l)}
=\left(\partial_x\ln\wt{\tau}_{n_\alpha,m_\beta}^{(l+1)}-\partial_x\ln\wt{\tau}_{n_\alpha,m_\beta}^{(l)}\right)\wt{\phi}_{n_\alpha,m_\beta}^{(l)}
\end{align*}
which is evidently satisfied by virtue of the continuous \ac{DT} of $\tau$ in \eqref{dFX:tau}.
\end{proof}

\subsection{N-fold Darboux transformations} It is evidence that the \ac{DT}s in Proposition \ref{dFX:1DT} may be repeated an arbitrary number of times. For the second step of this process we have
\begin{align}
\begin{split}\label{dFX:2DT-3}
{\phi[2]}_{n_\alpha,m_\beta}^{\left(l\right)}=~&\ld\wt{\phi}_{n_\alpha,m_\beta}^{\left(l+\alpha+\beta+1\right)}-\ld_2\wt{\theta}_{n_\alpha,m_\beta}^{\left(l+\alpha+\beta+1,2\right)}
{{\widetilde{\theta}}_{n_\alpha,m_\beta}^{{\left(l+\alpha+\beta,2\right)}^{-1}}}\wt{\phi}_{n_\alpha,m_\beta}^{\left(l+\alpha+\beta\right)}\\
=~&\partial_x\wt{\phi}_{n_\alpha,m_\beta}^{(l+\alpha+\beta)}-\partial_x\wt{\theta}_{n_\alpha,m_\beta}^{(l+\alpha+\beta,2)}
{{\widetilde{\theta}}_{n_\alpha,m_\beta}^{{\left(l+\alpha+\beta,2\right)}^{-1}}}\wt{\phi}_{n_\alpha,m_\beta}^{\left(l+\alpha+\beta\right)},
\end{split}\\
{\tau[2]}_{n_\alpha,m_\beta}^{(l)}=~&\wt{\theta}_{n_\alpha,m_\beta}^{\left(l+\alpha+\beta,2\right)}\wt\tau_{n_\alpha,m_\beta}^{(l+\alpha+\beta)},\label{dFX:tau2}
\end{align}
where ${\wt\theta}_{n_\alpha,m_\beta}^{\left(l,2\right)}$ is a non-zero fixed solution of the linear equations \eqref{dFX:Lax-tau} in $\wt{~}$  version  (i.e. for ${\wt\phi}_{n_\alpha,m_\beta}^{(l)}$ and ${\wt\tau}_{n_\alpha,m_\beta}^{(l)}$) with $\ld=\ld_2$, for $l=0, 1, \cdots, \cN-1$, as follows
\begin{align}\label{dFX:2DT-4}
\begin{split}
{\wt\theta}_{n_\alpha,m_\beta}^{\left(l,2\right)}=~&\ld_2\theta_{n_\alpha,m_\beta}^{\left(l+\alpha+\beta+1,2\right)}-\ld_1\theta_{n_\alpha,m_\beta}^{\left(l+\alpha+\beta+1,1\right)}
{\theta_{n_\alpha,m_\beta}^{\left(l+\alpha+\beta,1\right)}}^{-1}\theta_{n_\alpha,m_\beta}^{\left(l+\alpha+\beta,2\right)}\\
=~&\partial_x\theta_{n_\alpha,m_\beta}^{(l+\alpha+\beta,2)}-\partial_x\theta_{n_\alpha,m_\beta}^{(l+\alpha+\beta,1)}
{\theta_{n_\alpha,m_\beta}^{\left(l+\alpha+\beta,1\right)}}^{-1}\theta_{n_\alpha,m_\beta}^{\left(l+\alpha+\beta,2\right)},
\end{split}
\end{align}
generated by some fixed solution ${\theta}_{n_\alpha,m_\beta}^{\left(l,2\right)}$ of the linear equations \eqref{dFX:Lax-tau} for given $\tau_{n_\alpha,m_\beta}^{(l)}$ with $\ld=\ld_2$. Thus it can be seen that the \ac{DT}s may be repeated $N-$times and the solutions can be expressed completely in terms of the solutions of the initial equations \eqref{dFX:Lax-tau} without use of the solutions of any intermediate iterations of the procedure.

\begin{proposition}
Let the non-zero function $\theta_{n_\alpha,m_\beta}^{\left(l,i\right)}$ satisfies the linear equations \eqref{dFX:Lax-tau} for given $\tau_{n_\alpha,m_\beta}^{(l)}$ with $\ld=\ld_i$,  for $l=0, 1, \cdots, \cN-1$, $i=1, 2, \dots, N$. Then the so called N-fold Darboux transformations (\ac {DT}s) given as follows

\begin{align}
\begin{split}\label{dFX:NDT-1}
\phi_{n_\alpha,m_\beta}^{(l)} \rightarrow {\wt\phi[N]}_{n_\alpha,m_\beta}^{(l)}&=\frac{G_{n_\alpha,m_\beta}^{(l,N)}}{C_{n_\alpha,m_\beta}^{(l,N)}}\\
&=\frac{W(\theta_{n_\alpha,m_\beta}^{\left(l+N(\alpha+\beta),1\right)}, \dots, \theta_{n_\alpha,m_\beta}^{\left(l+N(\alpha+\beta),N\right)}, \phi_{n_\alpha,m_\beta}^{\left(l+N(\alpha+\beta)\right)} )}{W(\theta_{n_\alpha,m_\beta}^{\left(l+N(\alpha+\beta),1\right)}, \dots, \theta_{n_\alpha,m_\beta}^{\left(l+N(\alpha+\beta),N\right)})}
\end{split}
\\
\begin{split}\label{dFX:NDT-2}
\tau_{n_\alpha,m_\beta}^{(l)} \rightarrow  \wt\tau[N]_{n_\alpha,m_\beta}^{(l)}&=C_{n_\alpha,m_\beta}^{(l,N)}\cdot\tau_{n_\alpha,m_\beta}^{(l+N(\alpha+\beta))}\\
&=W(\theta_{n_\alpha,m_\beta}^{\left(l+N(\alpha+\beta),1\right)}, \dots, \theta_{n_\alpha,m_\beta}^{\left(l+N(\alpha+\beta),N\right)})\cdot\tau_{n_\alpha,m_\beta}^{(l+N(\alpha+\beta))}
\end{split}
\end{align}
for $l=0, 1, \cdots, \cN-1$, leaves the linear system \eqref{dFX:Lax-tau} invariant, where the determinants are defined as follows
\begin{align}
G_{n_\alpha,m_\beta}^{(l,N)}&=\begin{vmatrix}
{\theta}_{n_\alpha,m_\beta}^{\left(l\!+\!N(\alpha\!+\!\beta),1\right)}&\ld_1{\theta}_{n_\alpha,m_\beta}^{\left(l\!+\!N(\alpha\!+\!\beta)\!+\!1,1\right)}&\dots & \ld_1^N\theta_{n_\alpha,m_\beta}^{(l\!+\!N(\alpha\!+\!\beta)\!+\!N,1)}\\
{\theta}_{n_\alpha,m_\beta}^{\left(l\!+\!N(\alpha\!+\!\beta),2\right)}&\ld_2{\theta}_{n_\alpha,m_\beta}^{\left(l\!+\!N(\alpha+\beta)\!+\!1,2\right)}&\dots & \ld_2^N\theta_{n_\alpha,m_\beta}^{(l\!+\!N(\alpha\!+\!\beta)\!+\!N,2)}\\
\vdots&\vdots&\cdots&\vdots\\
{\theta}_{n_\alpha,m_\beta}^{\left(l\!+\!N(\alpha\!+\!\beta),N\right)}&\ld_N{\theta}_{n_\alpha,m_\beta}^{\left(l\!+\!N(\alpha\!+\!\beta)\!+\!1,N\right)}&\dots & \ld_N^N\theta_{n_\alpha,m_\beta}^{(l\!+\!N(\alpha\!+\!\beta)\!+\!N,N)}\\
\phi_{n_\alpha,m_\beta}^{(l\!+\!N(\alpha\!+\!\beta))}&\lambda\phi_{n_\alpha,m_\beta}^{(l\!+\!N(\alpha\!+\!\beta)\!+\!1))}& \dots & \lambda^N\phi_{n_\alpha,m_\beta}^{(l+N(\alpha+\beta)+N))}\\
\end{vmatrix},\label{caso-dete-1}\\
C_{n_\alpha,m_\beta}^{(l,N)}&=\begin{vmatrix}
{\theta}_{n_\alpha,m_\beta}^{\left(l\!+\!N(\alpha\!+\!\beta),1\right)}&\ld_1{\theta}_{n_\alpha,m_\beta}^{\left(l\!+\!N(\alpha\!+\!\beta)\!+\!1,1\right)}&\dots & \ld_1^{N\!-\!1}\theta_{n_\alpha,m_\beta}^{(l\!+\!N(\alpha\!+\!\beta)\!+\!(N\!-\!1),1)}\\
{\theta}_{n_\alpha,m_\beta}^{\left(l\!+\!N(\alpha\!+\!\beta),2\right)}&\ld_2{\theta}_{n_\alpha,m_\beta}^{\left(l\!+\!N(\alpha\!+\!\beta)\!+\!1,2\right)}&\dots & \ld_2^{N\!-\!1}\theta_{n_\alpha,m_\beta}^{(l\!+\!N(\alpha\!+\!\beta)\!+\!(N\!-\!1),2)}\\
\vdots&\vdots&\cdots&\vdots\\
{\theta}_{n_\alpha,m_\beta}^{\left(l\!+\!N(\alpha\!+\!\beta),N\right)}&\ld_N{\theta}_{n_\alpha,m_\beta}^{\left(l\!+\!N(\alpha\!+\!\beta)\!+\!1,N\right)}&\dots & \ld_N^{N\!-\!1}\theta_{n_\alpha,m_\beta}^{(l\!+\!N(\alpha\!+\!\beta)\!+\!(N\!-\!1),N)}\label{caso-dete-2}
\end{vmatrix},
\end{align}
$W(\theta_{n_\alpha,m_\beta}^{\left(l+N(\alpha+\beta),1\right)}, \dots, \theta_{n_\alpha,m_\beta}^{\left(l+N(\alpha+\beta),N\right)}, \phi_{n_\alpha,m_\beta}^{\left(l+N(\alpha+\beta)\right)} )$ and $W(\theta_{n_\alpha,m_\beta}^{\left(l+N(\alpha+\beta),1\right)}, \dots, \theta_{n_\alpha,m_\beta}^{\left(l+N(\alpha+\beta),N\right)})$ are Wronskian determinants\footnote{We define the Wronskian determinant $W$ of $k$ functions $\theta_1, \theta_2, \dots, \theta_k$ by
\begin{align*}
W(\theta_1, \theta_2, \dots, \theta_k)=det A, A_{ij}=\partial_x^{i-1}\theta_j, i,j=1,2,\dots,k.
\end{align*}
 }.
\end{proposition}

\begin{proof}
To verify that the first equation in the linear system \eqref{dFX:Lax-tau} is invariant when replace $\phi_{n_\alpha,m_\beta}^{(l)}$ by ${\wt\phi[N]}_{n_\alpha,m_\beta}^{(l)}$, we equivalently prove that
\begin{align*}
G_{n_\alpha+1,m_\beta}^{(l,N)}C_{n_\alpha,m_\beta}^{(l+\alpha+1,N)}=a_1\frac{\tau_{n_\alpha,m_\beta}^{(l+N(\alpha+\beta)+\alpha)}\tau_{n_\alpha+1,m_\beta}^{(l+N(\alpha+\beta)+1)}}
 {\tau_{n_\alpha,m_\beta}^{(l+N(\alpha+\beta)+\alpha+1)}\tau_{n_\alpha+1,m_\beta}^{(l+N(\alpha+\beta))}}G_{n_\alpha,m_\beta}^{(l+\alpha,N)}C_{n_\alpha+1,m_\beta}^{(l+1,N)}+\lambda G_{n_\alpha,m_\beta}^{(l+\alpha+1,N)}C_{n_\alpha+1,m_\beta}^{(l,N)},
\end{align*}
namely,
\begin{align}\label{Laplace}
G_{n_\alpha,m_\beta}^{(l-1,N)}C_{n_\alpha-1,m_\beta}^{(l+\alpha,N)}-a_1\frac{\tau_{n_\alpha-1,m_\beta}^{(l+N(\alpha+\beta)+\alpha-1)}\tau_{n_\alpha,m_\beta}^{(l+N(\alpha+\beta))}}
 {\tau_{n_\alpha-1,m_\beta}^{(l+N(\alpha+\beta)+\alpha)}\tau_{n_\alpha,m_\beta}^{(l+N(\alpha+\beta)-1)}}G_{n_\alpha-1,m_\beta}^{(l+\alpha-1,N)}C_{n_\alpha,m_\beta}^{(l,N)}-\lambda G_{n_\alpha-1,m_\beta}^{(l+\alpha,N)}C_{n_\alpha,m_\beta}^{(l-1,N)}=0,
\end{align}
by moving the component index $l\rightarrow l-1$ and shifting backward the discrete variable $n_\alpha\rightarrow n_\alpha-1$.

To prove \eqref{Laplace} is true, the way is using the Laplace theorem of a $(2N+1)\times(2N+1)$ determinant. From the definition \eqref{caso-dete-1} and \eqref{caso-dete-2}, we can easily write out $G_{n_\alpha,m_\beta}^{(l-1,N)}$ and $C_{n_\alpha,m_\beta}^{(l-1,N)}$, namely
\begin{align}
G_{n_\alpha,m_\beta}^{(l-1,N)}&=\begin{vmatrix}
{\theta}_{n_\alpha,m_\beta}^{\left(l\!+\!N(\alpha\!+\!\beta)-1,1\right)}&\ld_1{\theta}_{n_\alpha,m_\beta}^{\left(l\!+\!N(\alpha\!+\!\beta),1\right)}&\dots & \ld_1^N\theta_{n_\alpha,m_\beta}^{(l\!+\!N(\alpha\!+\!\beta)\!+\!N-1,1)}\\
{\theta}_{n_\alpha,m_\beta}^{\left(l\!+\!N(\alpha\!+\!\beta)-1,2\right)}&\ld_2{\theta}_{n_\alpha,m_\beta}^{\left(l\!+\!N(\alpha+\beta),2\right)}&\dots & \ld_2^N\theta_{n_\alpha,m_\beta}^{(l\!+\!N(\alpha\!+\!\beta)\!+\!N-1,2)}\\
\vdots&\vdots&\cdots&\vdots\\
{\theta}_{n_\alpha,m_\beta}^{\left(l\!+\!N(\alpha\!+\!\beta)-1,N\right)}&\ld_N{\theta}_{n_\alpha,m_\beta}^{\left(l\!+\!N(\alpha\!+\!\beta),N\right)}&\dots & \ld_N^N\theta_{n_\alpha,m_\beta}^{(l\!+\!N(\alpha\!+\!\beta)\!+\!N-1,N)}\\
\phi_{n_\alpha,m_\beta}^{(l\!+\!N(\alpha\!+\!\beta)-1)}&\lambda\phi_{n_\alpha,m_\beta}^{(l\!+\!N(\alpha\!+\!\beta)))}& \dots & \lambda^N\phi_{n_\alpha,m_\beta}^{(l+N(\alpha+\beta)+N-1)}\\
\end{vmatrix},\label{caso-dete-3}\\
 \mathop{\Pi}_{i=1}^{N}\ld_iC_{n_\alpha,m_\beta}^{(l-1,N)}&=\begin{vmatrix}
\ld_1^{-1}{\theta}_{n_\alpha,m_\beta}^{\left(l\!+\!N(\alpha\!+\!\beta)-1,1\right)}&{\theta}_{n_\alpha,m_\beta}^{\left(l\!+\!N(\alpha\!+\!\beta),1\right)}&\dots & \ld_1^{N\!-\!2}\theta_{n_\alpha,m_\beta}^{(l\!+\!N(\alpha\!+\!\beta)\!+\!(N\!-\!2),1)}\\
\ld_2^{-1}{\theta}_{n_\alpha,m_\beta}^{\left(l\!+\!N(\alpha\!+\!\beta)-1,2\right)}&{\theta}_{n_\alpha,m_\beta}^{\left(l\!+\!N(\alpha\!+\!\beta,2\right)}&\dots & \ld_2^{N\!-\!2}\theta_{n_\alpha,m_\beta}^{(l\!+\!N(\alpha\!+\!\beta)\!+\!(N\!-\!2),2)}\\
\vdots&\vdots&\cdots&\vdots\\
\ld_N^{-1}{\theta}_{n_\alpha,m_\beta}^{\left(l\!+\!N(\alpha\!+\!\beta)-1,N\right)}&{\theta}_{n_\alpha,m_\beta}^{\left(l\!+\!N(\alpha\!+\!\beta),N\right)}&\dots & \ld_N^{N\!-\!2}\theta_{n_\alpha,m_\beta}^{(l\!+\!N(\alpha\!+\!\beta)\!+\!(N\!-\!2),N)}\label{caso-dete-4}
\end{vmatrix}.
\end{align}

As for the determinant expression of $C_{n_\alpha-1,m_\beta}^{(l+\alpha,N)}$,
$a_1\frac{\tau_{n_\alpha-1,m_\beta}^{(l+N(\alpha+\beta)+\alpha-1)}\tau_{n_\alpha,m_\beta}^{(l+N(\alpha+\beta))}}{\tau_{n_\alpha-1,m_\beta}^{(l+N(\alpha+\beta)+\alpha)}\tau_{n_\alpha,m_\beta}^{(l+N(\alpha+\beta)-1)}}G_{n_\alpha-1,m_\beta}^{(l+\alpha-1,N)}$
and $G_{n_\alpha-1,m_\beta}^{(l+\alpha,N)}$, we need to use the linear equation \eqref{dFX:Lax1-tau} to express them in another way respectively.

Taking into account that the linear equation \eqref{dFX:Lax1-tau} when $l \rightarrow l+N(\alpha+\beta)+s-2$, $s= N, N-1, \dots ,2$, and  shifting $n_\alpha\rightarrow n_\alpha-1$, we have
\begin{align*}
\phi_{n_\alpha,m_\beta}^{(l+N(\alpha+\beta)+s-2)}&=a_1\frac{\tau_{n_\alpha-1,m_\beta}^{(l+N(\alpha+\beta)+s-2+\alpha)}\tau_{n_\alpha,m_\beta}^{(l+N(\alpha+\beta)+s-1)}}
 {\tau_{n_\alpha-1,m_\beta}^{(l+N(\alpha+\beta)+s-1+\alpha)}\tau_{n_\alpha,m_\beta}^{(l+N(\alpha+\beta)+s-2)}}
 \phi_{n_\alpha-1,m_\beta}^{(l+N(\alpha+\beta)+s-2+\alpha)}+\ld\phi_{n_\alpha-1,m_\beta}^{(l+N(\alpha+\beta)+s-1+\alpha)},
\end{align*}
namely
\begin{align}\label{Laplace-1}
\ld^{s-1}\phi_{n_\alpha-1,m_\beta}^{(l+N(\alpha+\beta)+s-1+\alpha)}=\ld^{s-2}\bigg[\phi_{n_\alpha,m_\beta}^{(l+N(\alpha+\beta)+s-2)}-a_1\frac{\tau_{n_\alpha-1,m_\beta}^{(l+N(\alpha+\beta)+s-2+\alpha)}\tau_{n_\alpha,m_\beta}^{(l+N(\alpha+\beta)+s-1)}}
 {\tau_{n_\alpha-1,m_\beta}^{(l+N(\alpha+\beta)+s-1+\alpha)}\tau_{n_\alpha,m_\beta}^{(l+N(\alpha+\beta)+s-2)}}
 \phi_{n_\alpha-1,m_\beta}^{(l+N(\alpha+\beta)+s-2+\alpha)}\bigg].
\end{align}
Similarly, when $\ld=\ld_i$, $i=1, 2, \dots, N$, we have
\begin{align}\label{Laplace-2}
\ld_i^{s-1}\theta_{n_\alpha-1,m_\beta}^{(l+N(\alpha+\beta)+s-1+\alpha,i)}=\ld_i^{s-2}\bigg[\theta_{n_\alpha,m_\beta}^{(l+N(\alpha+\beta)+s-2,i)}
-a_1\frac{\tau_{n_\alpha-1,m_\beta}^{(l+N(\alpha+\beta)+s-2+\alpha)}\tau_{n_\alpha,m_\beta}^{(l+N(\alpha+\beta)+s-1)}}
 {\tau_{n_\alpha-1,m_\beta}^{(l+N(\alpha+\beta)+s-1+\alpha)}\tau_{n_\alpha,m_\beta}^{(l+N(\alpha+\beta)+s-2)}}
 \theta_{n_\alpha-1,m_\beta}^{(l+N(\alpha+\beta)+s-2+\alpha,i)}\bigg].
\end{align}

Substituting \eqref{Laplace-2} for $s=N, N-1, \dots, 2$ respectively into the columns, in reverse order, from the $N{th}$ to the second, we obtain
\begin{align}\label{Laplace-7}
C_{n_\alpha-1,m_\beta}^{(l+\alpha,N)}&=\begin{vmatrix}
{\theta}_{n_\alpha-1,m_\beta}^{\left(l\!+\!N(\alpha\!+\!\beta)+\alpha,1\right)}&{\theta}_{n_\alpha,m_\beta}^{\left(l\!+\!N(\alpha\!+\!\beta),1\right)}&\dots & \ld_1^{N\!-\!2}\theta_{n_\alpha,m_\beta}^{(l\!+\!N(\alpha\!+\!\beta)\!+\!(N\!-\!2),1)}\\
{\theta}_{n_\alpha-1,m_\beta}^{\left(l\!+\!N(\alpha\!+\!\beta)+\alpha,2\right)}&{\theta}_{n_\alpha,m_\beta}^{\left(l\!+\!N(\alpha\!+\!\beta),2\right)}&\dots & \ld_2^{N\!-\!2}\theta_{n_\alpha,m_\beta}^{(l\!+\!N(\alpha\!+\!\beta)\!+\!(N\!-\!2),2)}\\
\vdots&\vdots&\cdots&\vdots\\
{\theta}_{n_\alpha-1,m_\beta}^{\left(l\!+\!N(\alpha\!+\!\beta)+\alpha,N\right)}&{\theta}_{n_\alpha,m_\beta}^{\left(l\!+\!N(\alpha\!+\!\beta),N\right)}&\dots & \ld_N^{N\!-\!2}\theta_{n_\alpha,m_\beta}^{(l\!+\!N(\alpha\!+\!\beta)\!+\!(N\!-\!2),N)}
\end{vmatrix}.
\end{align}

Taking into account that the linear equation \eqref{dFX:Lax1-tau} when $l \rightarrow l+N(\alpha+\beta)+s-3$, $s=N+1, N,\dots ,2$, and  shifting $n_\alpha\rightarrow n_\alpha-1$, we have
\begin{align*}
\phi_{n_\alpha,m_\beta}^{(l+N(\alpha+\beta)+s-3)}&=a_1\frac{\tau_{n_\alpha-1,m_\beta}^{(l+N(\alpha+\beta)+s-3+\alpha)}\tau_{n_\alpha,m_\beta}^{(l+N(\alpha+\beta)+s-2)}}
 {\tau_{n_\alpha-1,m_\beta}^{(l+N(\alpha+\beta)+s-2+\alpha)}\tau_{n_\alpha,m_\beta}^{(l+N(\alpha+\beta)+s-3)}}
 \phi_{n_\alpha-1,m_\beta}^{(l+N(\alpha+\beta)+s-3+\alpha)}+\ld\phi_{n_\alpha-1,m_\beta}^{(l+N(\alpha+\beta)+s-2+\alpha)},
\end{align*}
namely
\begin{align}\label{Laplace-3}
\ld^{s-1}\phi_{n_\alpha-1,m_\beta}^{(l+N(\alpha+\beta)+s-2+\alpha)}=\ld^{s-2}\bigg[\phi_{n_\alpha,m_\beta}^{(l+N(\alpha+\beta)+s-3)}-a_1\frac{\tau_{n_\alpha-1,m_\beta}^{(l+N(\alpha+\beta)+s-3+\alpha)}\tau_{n_\alpha,m_\beta}^{(l+N(\alpha+\beta)+s-2)}}
 {\tau_{n_\alpha-1,m_\beta}^{(l+N(\alpha+\beta)+s-2+\alpha)}\tau_{n_\alpha,m_\beta}^{(l+N(\alpha+\beta)+s-3)}}
 \phi_{n_\alpha-1,m_\beta}^{(l+N(\alpha+\beta)+s-3+\alpha)}\bigg].
\end{align}
Similarly, when $\ld=\ld_i$, $i=1, 2, \dots, N$, we have
\begin{align}\label{Laplace-4}
\ld_i^{s-1}\theta_{n_\alpha-1,m_\beta}^{(l+N(\alpha+\beta)+s-2+\alpha,i)}=\ld_i^{s-2}\bigg[\theta_{n_\alpha,m_\beta}^{(l+N(\alpha+\beta)+s-3,i)}
-a_1\frac{\tau_{n_\alpha-1,m_\beta}^{(l+N(\alpha+\beta)+s-3+\alpha)}\tau_{n_\alpha,m_\beta}^{(l+N(\alpha+\beta)+s-2)}}
 {\tau_{n_\alpha-1,m_\beta}^{(l+N(\alpha+\beta)+s-2+\alpha)}\tau_{n_\alpha,m_\beta}^{(l+N(\alpha+\beta)+s-3)}}
 \theta_{n_\alpha-1,m_\beta}^{(l+N(\alpha+\beta)+s-3+\alpha,i)}\bigg].
\end{align}
Substituting \eqref{Laplace-3} and \eqref{Laplace-4} for $s=N+1, N, \dots, 2$ respectively into the columns, in reverse order, from the $(N+1){th}$ to the second, we obtain
\begin{align*}
G_{n_\alpha-1,m_\beta}^{(l+\alpha-1,N)}&=\begin{vmatrix}
{\theta}_{n_\alpha-1,m_\beta}^{\left(l\!+\!N(\alpha\!+\!\beta)+\alpha-1,1\right)}&{\theta}_{n_\alpha,m_\beta}^{\left(l\!+\!N(\alpha\!+\!\beta)-1,1\right)}&\dots & \ld_1^{N\!-\!1}\theta_{n_\alpha,m_\beta}^{(l\!+\!N(\alpha\!+\!\beta)\!+\!(N\!-\!2),1))}\\
{\theta}_{n_\alpha-1,m_\beta}^{\left(l\!+\!N(\alpha\!+\!\beta)+\alpha-1,2\right)}&{\theta}_{n_\alpha,m_\beta}^{\left(l\!+\!N(\alpha\!+\!\beta)-1,2\right)}&\dots & \ld_2^{N\!-\!1}\theta_{n_\alpha,m_\beta}^{(l\!+\!N(\alpha\!+\!\beta)\!+\!(N\!-\!2),2))}\\
\vdots&\vdots&\cdots&\vdots\\
{\theta}_{n_\alpha-1,m_\beta}^{\left(l\!+\!N(\alpha\!+\!\beta)+\alpha-1,N\right)}&{\theta}_{n_\alpha,m_\beta}^{\left(l\!+\!N(\alpha\!+\!\beta)-1,N\right)}&\dots & \ld_N^{N\!-\!1}\theta_{n_\alpha,m_\beta}^{(l\!+\!N(\alpha\!+\!\beta)\!+\!(N\!-\!2),N))}\\
{\phi}_{n_\alpha-1,m_\beta}^{\left(l\!+\!N(\alpha\!+\!\beta)+\alpha-1\right)}&{\phi}_{n_\alpha,m_\beta}^{\left(l\!+\!N(\alpha\!+\!\beta)-1\right)}&\dots & \ld^{N\!-\!1}\phi_{n_\alpha,m_\beta}^{(l\!+\!N(\alpha\!+\!\beta)\!+\!(N\!-\!2)))}
\end{vmatrix}.
\end{align*}
For the first column in the above determinant, with the help of \eqref{Laplace-3} and \eqref{Laplace-4} for $s=2$, i.e.,
\begin{align*}
a_1\frac{\tau_{n_\alpha-1,m_\beta}^{(l+N(\alpha+\beta)+\alpha-1)}\tau_{n_\alpha,m_\beta}^{(l+N(\alpha+\beta))}}
 {\tau_{n_\alpha-1,m_\beta}^{(l+N(\alpha+\beta)+\alpha)}\tau_{n_\alpha,m_\beta}^{(l+N(\alpha+\beta)-1)}}
 \phi_{n_\alpha-1,m_\beta}^{(l+N(\alpha+\beta)+\alpha-1)}=\phi_{n_\alpha,m_\beta}^{(l+N(\alpha+\beta)-1)}
-\ld\phi_{n_\alpha-1,m_\beta}^{(l+N(\alpha+\beta)+\alpha)},
\end{align*}
and
\begin{align*}
a_1\frac{\tau_{n_\alpha-1,m_\beta}^{(l+N(\alpha+\beta)+\alpha-1)}\tau_{n_\alpha,m_\beta}^{(l+N(\alpha+\beta))}}
 {\tau_{n_\alpha-1,m_\beta}^{(l+N(\alpha+\beta)+\alpha)}\tau_{n_\alpha,m_\beta}^{(l+N(\alpha+\beta)-1)}}
 \theta_{n_\alpha-1,m_\beta}^{(l+N(\alpha+\beta)+\alpha-1,i)}=\theta_{n_\alpha,m_\beta}^{(l+N(\alpha+\beta)-1,i)}
-\ld_i\theta_{n_\alpha-1,m_\beta}^{(l+N(\alpha+\beta)+\alpha,i)}, i=1, 2, \dots, N,
\end{align*}
we get
\begin{align}\label{Laplace-5}
a_1\frac{\tau_{n_\alpha-1,m_\beta}^{(l+N(\alpha+\beta)+\alpha-1)}\tau_{n_\alpha,m_\beta}^{(l+N(\alpha+\beta))}}
 {\tau_{n_\alpha-1,m_\beta}^{(l+N(\alpha+\beta)+\alpha)}\tau_{n_\alpha,m_\beta}^{(l+N(\alpha+\beta)-1)}}G_{n_\alpha-1,m_\beta}^{(l+\alpha-1,N)}&=-\begin{vmatrix}
\ld_1{\theta}_{n_\alpha-1,m_\beta}^{\left(l\!+\!N(\alpha\!+\!\beta)+\alpha,1\right)}&{\theta}_{n_\alpha,m_\beta}^{\left(l\!+\!N(\alpha\!+\!\beta)-1,1\right)}&\dots & \ld_1^{N\!-\!1}\theta_{n_\alpha,m_\beta}^{(l\!+\!N(\alpha\!+\!\beta)\!+\!(N\!-\!2),1))}\\
\ld_2{\theta}_{n_\alpha-1,m_\beta}^{\left(l\!+\!N(\alpha\!+\!\beta)+\alpha,2\right)}&{\theta}_{n_\alpha,m_\beta}^{\left(l\!+\!N(\alpha\!+\!\beta)-1,2\right)}&\dots & \ld_2^{N\!-\!1}\theta_{n_\alpha,m_\beta}^{(l\!+\!N(\alpha\!+\!\beta)\!+\!(N\!-\!2),2))}\\
\vdots&\vdots&\cdots&\vdots\\
\ld_N{\theta}_{n_\alpha-1,m_\beta}^{\left(l\!+\!N(\alpha\!+\!\beta)+\alpha,N\right)}&{\theta}_{n_\alpha,m_\beta}^{\left(l\!+\!N(\alpha\!+\!\beta)-1,N\right)}&\dots & \ld_N^{N\!-\!1}\theta_{n_\alpha,m_\beta}^{(l\!+\!N(\alpha\!+\!\beta)\!+\!(N\!-\!2),N))}\\
\ld{\phi}_{n_\alpha-1,m_\beta}^{\left(l\!+\!N(\alpha\!+\!\beta)+\alpha\right)}&{\phi}_{n_\alpha,m_\beta}^{\left(l\!+\!N(\alpha\!+\!\beta)-1\right)}&\dots & \ld^{N\!-\!1}\phi_{n_\alpha,m_\beta}^{(l\!+\!N(\alpha\!+\!\beta)\!+\!(N\!-\!2)))}
\end{vmatrix}.
\end{align}
Similarly, in the light of \eqref{Laplace-1} and \eqref{Laplace-2} for $s=N+1, N, \dots, 2$, we obtain
\begin{align}\label{Laplace-6}
\lambda G_{n_\alpha-1,m_\beta}^{(l+\alpha,N)}= \mathop{\Pi}_{i=1}^{N}\ld_i^{-1}\begin{vmatrix}
\ld_1{\theta}_{n_\alpha-1,m_\beta}^{\left(l\!+\!N(\alpha\!+\!\beta)+\alpha,1\right)}&\ld_1{\theta}_{n_\alpha,m_\beta}^{\left(l\!+\!N(\alpha\!+\!\beta),1\right)}&\dots & \ld_1^{N}\theta_{n_\alpha,m_\beta}^{(l\!+\!N(\alpha\!+\!\beta)\!+\!(N\!-\!1),1))}\\
\ld_2{\theta}_{n_\alpha-1,m_\beta}^{\left(l\!+\!N(\alpha\!+\!\beta)+\alpha,2\right)}&\ld_2{\theta}_{n_\alpha,m_\beta}^{\left(l\!+\!N(\alpha\!+\!\beta),2\right)}&\dots & \ld_2^{N}\theta_{n_\alpha,m_\beta}^{(l\!+\!N(\alpha\!+\!\beta)\!+\!(N\!-\!1),2))}\\
\vdots&\vdots&\cdots&\vdots\\
\ld_N{\theta}_{n_\alpha-1,m_\beta}^{\left(l\!+\!N(\alpha\!+\!\beta)+\alpha,N\right)}&\ld_N{\theta}_{n_\alpha,m_\beta}^{\left(l\!+\!N(\alpha\!+\!\beta),N\right)}&\dots & \ld_N^{N}\theta_{n_\alpha,m_\beta}^{(l\!+\!N(\alpha\!+\!\beta)\!+\!(N\!-\!1),N))}\\
\ld{\phi}_{n_\alpha-1,m_\beta}^{\left(l\!+\!N(\alpha\!+\!\beta)+\alpha\right)}&\ld{\phi}_{n_\alpha,m_\beta}^{\left(l\!+\!N(\alpha\!+\!\beta)\right)}&\dots & \ld^{N}\phi_{n_\alpha,m_\beta}^{(l\!+\!N(\alpha\!+\!\beta)\!+\!(N\!-\!1)))}
\end{vmatrix}.
\end{align}

Substituting those above determinants into the left-hand side of the equation \eqref{Laplace} gives the $(2N+1)\times(2N+1)$ vanishing determinant.

In a similar way, we can prove that the second equation \eqref{dFX:Lax2-tau} and \eqref{dFX:Lax3-tau}  are true. Otherwise, by using the continuous part in the linear system, we can see that the two determinants in the DT are equal.
\end{proof}

\section{Exact solutions}\label{S:Sol}

In this section, we present explicit examples of the solutions that may be obtained by means of the Darboux transformations.
We choose the seed solution of the $\mathbb{Z}_\cN$ graded discrete integrable models \eqref{dFX} as $\tau_{n_\alpha,m_\beta}^{(l)}=1$ (by \eqref{cFX:BLT} $u_{n_\alpha, m_\beta}^{(l)}=0$ and $v_{n_\alpha, m_\beta}^{(l)}=1$) for $l=0, 1, \dots, \cN-1$. With this choice, the $\mathbb{Z}_\cN$ graded discrete Lax pairs \eqref{dFX:Lax-tau} read as follows
\begin{align*}
&\phi_{n_\alpha+1,m_\beta}^{(l)}=a_1\phi_{n_\alpha,m_\beta}^{(l+\alpha)}+\ld\phi_{n_\alpha,m_\beta}^{(l+\alpha+1)}, \\
&\phi_{n_\alpha,m_\beta+1}^{(l)}=a_2\phi_{n_\alpha,m_\beta}^{(l+\beta)}+\ld\phi_{n_\alpha,m_\beta}^{(l+\beta+1)}, \\
&\partial_x\phi_{n_\alpha,m_\beta}^{(l)}=\ld\phi_{n_\alpha,m_\beta}^{(l+1)}.
\end{align*}
Its $N$ linearly independent nonzero solutions, depending on parameter $\ld_i$ are found to be
\begin{align}
\theta^{(l,i)}_{n_\alpha,m_\beta}=\sum_{r=0}^{\cN-1}e^{\oa_\cN^{r} \ld_i x}{\oa_\cN^r}^{(l+\alpha n_\alpha+\beta m_\beta)} \left(a_1+\oa_\cN^r \ld_i\right)^{n_{\alpha}}\left(a_2+\oa_\cN^r \ld_i\right)^{m_\beta}, \quad i=1,2,\dots,N,\label{theta-i}
\end{align}
where the complex number $\oa_\cN^{r}=e^{\frac{2 r\pi \mathfrak{i}}{\cN}}$ , $l=0,1,\dots,\cN-1$. So the N-soliton solution of the $\mathbb{Z}_\cN$ graded discrete integrable models can be obtained by the N-fold \ac{DT}s  in \eqref{dFX:NDT-2} with elements $\theta^{(l,i)}_{n_\alpha,m_\beta}$ given in \eqref{theta-i}.

For instance, when $N=1$, $\cN=2$, and $(\alpha, \beta)=(0,1)$, we have the 1-soliton solution
\begin{align*}
v^{(0)}_{n_0,m_1}=\frac{e^{\lambda_1x}+e^{-\lambda_1x}\rho_{n_0,m_1}^{(1)}}{e^{\ld_1x}+e^{-\lambda_1x}\rho_{n_0,m_1}^{(0)}}, \quad v^{(1)}_{n_0,m_1}=\frac{1}{v^{(0)}_{n_0,m_1}}
\end{align*}
of the discrete-time \ac{2DTL} of $A_1^{(1)}$-type; where $\rho_{n_0,m_1}^{(l)}=\left(\frac{a_1-\ld_1}{a_1+\ld_1}\right)^{n_0}\!\!\left(\frac{a_2-\ld_1}{a_2+\ld_1}\right)^{m_1}\!\!(-1)^l$ , $l=0, 1$, is the discrete plane wave factor. Moreover, when $N=1$, $\cN=3$ and $(\alpha, \beta)=(0,2)$, we have the 1-soliton solution
\begin{align*}
&v^{(0)}_{n_0,m_2}=\frac{e^{\lambda_1x}+e^{\oa_3^{1}\lambda_1x}\rho_{n_0,m_2}^{(1,\oa_3^{1})}+e^{\oa_3^{2}\lambda_1x}\rho_{n_0,m_2}^{(1,\oa_3^{2})}}{e^{\lambda_1x}+e^{\oa_3^{1}\lambda_1x}\rho_{n_0,m_2}^{(0,\oa_3^{1})}+e^{\oa_3^{2}\lambda_1x}\rho_{n_0,m_2}^{(0,\oa_3^{2})}},\\
& v^{(1)}_{n_0,m_2}=\frac{e^{\lambda_1x}+e^{\oa_3^{1}\lambda_1x}\rho_{n_0,m_2}^{(2,\oa_3^{1})}+e^{\oa_3^{2}\lambda_1x}\rho_{n_0,m_2}^{(2,\oa_3^{2})}}{e^{\lambda_1x}+e^{\oa_3^{1}\lambda_1x}\rho_{n_0,m_2}^{(1,\oa_3^{1})}+e^{\oa_3^{2}\lambda_1x}\rho_{n_0,m_2}^{(1,\oa_3^{2})}},\\
 &v^{(2)}_{n_0,m_2}=\frac{1}{v^{(0)}_{n_0,m_2}v^{(1)}_{n_0,m_2}}
 \end{align*}
of the discrete-time \ac{2DTL} of $A_2^{(1)}$-type; where $\rho_{n_0,m_2}^{(l,\oa_3^{r})}=\left(\frac{a_1-\ld_1}{a_1+\ld_1}\right)^{n_0}\!\!\left(\frac{a_2-\ld_1}{a_2+\ld_1}\right)^{m_2}\!\!(\oa_3^{r})^l$ , $\oa_3^{r}=e^{\frac{2 r\pi \mathfrak{i}}{3}}$, $l=0, 1, 2$, $r=0,1,2$ is the discrete plane wave factor.

\section{Conclusions}\label{S:Concl}

In this paper, we considered the Darboux transformations of the $\mathbb{Z}_\cN$ graded discrete integrable systems (FX models in coprime case), in order to present their multi-soliton solution as explicit expressions.
To our knowledge  \cite{SNZ14} and motivated by the universality of the discrete KP-type equations, we were naturally led to considering reductions of the discrete KP-type equations in order to recover multi-component discrete systems ( with the idea of increasing by one the number of components in order to decrease by one the dimension).
By periodic reductions, we in our scheme obtained the discrete Gel'fand-Dikii hierarchy \eqref{dGD} associated with their Lax pairs \eqref{dGD:Lax} and \eqref{dGD:Adjoint} from the discrete KP-type equations.
Then by introducing a variable transformations \eqref{dFX:CoV}, we recovered not only the \ac{FX} discrete systems in coprime case, the corresponding $\mathbb{Z}_\cN$-graded Lax pairs
and the novel adjoint discrete Lax pairs, but also their compatible differential-difference analogues and the bilinear formalism of these equations.

Another remark is that we gave three forms of the discrete Gel'fand-Dikii hierarchy, which are unmodifed, modified and bilinear forms, and the  Miura maps between them.
Different to the continuous case, one cannot construct solutions of a nonlinear equation directly from that of another nonlinear equation related by the discrete Miura maps.
So we introduce extra variables, which can be either continuous or discrete, and overcome the disadvantage (namely nonlocality) bringing by the discrete Miura maps, see \eqref{cFX:BLT}.

We only show the multi-soliton solutions of the $\mathbb{Z}_\cN$ graded discrete integrable systems in this paper.
It is quite certain that other classes of solutions can also be obtained under the framework of Darboux transformations.
On the construction of Darboux transformations, here they were obtained through reduction and consistency with the semi-discrete parts.
In fact, it is interesting to consider how to construct the Darboux transformation by factorisation of difference operators like in the continuous case.
The exact solutions of the non-coprime case of the FX models are still unclear.
These above questions are our future work.

\section*{Acknowledgments}

This project was supported by the National Natural Science Foundation of China (NSFC) under grant 11501510.
Part of the author's work was carried out during her PhD study at the University of Glasgow,
and she was very grateful to her supervisor Jonathan Nimmo for his supervision and many stimulating discussions on the theory of Darboux transformations.

\appendix

\section{Examples for $\cN=2$ and $\cN=3$}\label{S:Exampl}

\subsection{$\cN=2$}

\subsubsection{$(\alpha,\beta)=(0,0)$}

\begin{itemize}

\item{Additive potential}

\begin{itemize}

\item{First integrals}
\begin{align*}
 &\left(a_1+u_{n_0,m_0}^{(1)}-u_{n_0+1,m_0}^{(0)}\right)\left(a_1+u_{n_0,m_0}^{(0)}-u_{n_0+1,m_0}^{(1)}\right)=a_1^2, \\
 &\left(a_2+u_{n_0,m_0}^{(1)}-u_{n_0,m_0+1}^{(0)}\right)\left(a_2+u_{n_0,m_0}^{(0)}-u_{n_0,m_0+1}^{(1)}\right)=a_2^2;
\end{align*}
\item{Lax matrices}
\begin{align*}
 &\bL_{n_0,m_0}=
 \begin{pmatrix}
  a_1+u_{n_0,m_0}^{(1)}-u_{n_0+1,m_0}^{(0)} & \lambda \\
  \lambda & a_1+u_{n_0,m_0}^{(0)}-u_{n_0+1,m_0}^{(1)}
 \end{pmatrix}, \\
 &\bM_{n_0,m_0}=
 \begin{pmatrix}
  a_2+u_{n_0,m_0}^{(1)}-u_{n_0,m_0+1}^{(0)} & \lambda \\
  \lambda & a_2+u_{n_0,m_0}^{(0)}-u_{n_0,m_0+1}^{(1)}
 \end{pmatrix};
\end{align*}
\item{Two-component system}
\begin{align*}
 &\frac{a_1+u_{n_0,m_0+1}^{(1)}-u_{n_0+1,m_0+1}^{(0)}}{a_1+u_{n_0,m_0}^{(1)}-u_{n_0+1,m_0}^{(0)}}
 =\frac{a_2+u_{n_0+1,m_0}^{(1)}-u_{n_0+1,m_0+1}^{(0)}}{a_2+u_{n_0,m_0}^{(1)}-u_{n_0,m_0+1}^{(0)}}, \\
 &\frac{a_1+u_{n_0,m_0+1}^{(0)}-u_{n_0+1,m_0+1}^{(1)}}{a_1+u_{n_0,m_0}^{(0)}-u_{n_0+1,m_0}^{(1)}}
 =\frac{a_2+u_{n_0+1,m_0}^{(0)}-u_{n_0+1,m_0+1}^{(1)}}{a_2+u_{n_0,m_0}^{(0)}-u_{n_0,m_0+1}^{(1)}};
\end{align*}
\item{Scalar equation (discrete (unmodified) \ac{KdV} equation, cf. \eqref{dKdV:u})}
\begin{align*}
 \left(a_1+a_2+u_{n_0,m_0}^{(0)}-u_{n_0+1,m_0+1}^{(0)}\right)
 \left(a_1-a_2+u_{n_0,m_0+1}^{(0)}-u_{n_0+1,m_0}^{(0)}\right)=a_1^2-a_2^2.
\end{align*}

\end{itemize}

\item{Quotient potential}

\begin{itemize}

\item{First integrals}
\begin{align*}
 v_{n_0,m_0}^{(0)}v_{n_0,m_0}^{(1)}=1;
\end{align*}
\item{Lax matrices}
\begin{align*}
 \bL_{n_0,m_0}=
 \begin{pmatrix}
  a_1\frac{v_{n_0+1,m_0}^{(0)}}{v_{n_0,m_0}^{(0)}} & \lambda \\
  \lambda & a_1\frac{v_{n_0+1,m_0}^{(1)}}{v_{n_0,m_0}^{(1)}}
 \end{pmatrix}, \quad
 \bM_{n_0,m_0}=
 \begin{pmatrix}
  a_2\frac{v_{n_0,m_0+1}^{(0)}}{v_{n_0,m_0}^{(0)}} & \lambda \\
  \lambda & a_2\frac{v_{n_0,m_0+1}^{(1)}}{v_{n_0,m_0}^{(1)}}
 \end{pmatrix};
\end{align*}
\item{Two-component system}
\begin{align*}
 &a_1\left(\frac{v_{n_0+1,m_0+1}^{(0)}}{v_{n_0,m_0+1}^{(0)}}-\frac{v_{n_0+1,m_0}^{(1)}}{v_{n_0,m_0}^{(1)}}\right)
 =a_2\left(\frac{v_{n_0+1,m_0+1}^{(0)}}{v_{n_0+1,m_0}^{(0)}}-\frac{v_{n_0,m_0+1}^{(1)}}{v_{n_0,m_0}^{(1)}}\right), \\
 &a_1\left(\frac{v_{n_0+1,m_0+1}^{(1)}}{v_{n_0,m_0+1}^{(1)}}-\frac{v_{n_0+1,m_0}^{(0)}}{v_{n_0,m_0}^{(0)}}\right)
 =a_2\left(\frac{v_{n_0+1,m_0+1}^{(1)}}{v_{n_0+1,m_0}^{(1)}}-\frac{v_{n_0,m_0+1}^{(0)}}{v_{n_0,m_0}^{(0)}}\right);
\end{align*}
\item{Scalar equation (discrete modified \ac{KdV} equation, cf. \eqref{dKdV:v})}
\begin{align*}
 a_1\left(v_{n_0,m_0}^{(0)}v_{n_0,m_0+1}^{(0)}-v_{n_0+1,m_0}^{(0)}v_{n_0+1,m_0+1}^{(0)}\right)
 =a_2\left(v_{n_0,m_0}^{(0)}v_{n_0+1,m_0}^{(0)}-v_{n_0,m_0+1}^{(0)}v_{n_0+1,m_0+1}^{(0)}\right).
\end{align*}

\end{itemize}

\item{Bilinear potential}

\begin{itemize}

\item{Lax matrices}
\begin{align*}
 &\bL_{n_0,m_0}=
 \begin{pmatrix}
  a_1\frac{\tau_{n_0+1,m_0}^{(1)}\tau_{n_0,m_0}^{(0)}}{\tau_{n_0+1,m_0}^{(0)}\tau_{n_0,m_0}^{(1)}} & \lambda \\
  \lambda & a_1\frac{\tau_{n_0+1,m_0}^{(0)}\tau_{n_0,m_0}^{(1)}}{\tau_{n_0+1,m_0}^{(1)}\tau_{n_0,m_0}^{(0)}}
 \end{pmatrix}, \\
 &\bM_{n_0,m_0}=
 \begin{pmatrix}
  a_2\frac{\tau_{n_0,m_0+1}^{(1)}\tau_{n_0,m_0}^{(0)}}{\tau_{n_0,m_0+1}^{(0)}\tau_{n_0,m_0}^{(1)}} & \lambda \\
  \lambda & a_2\frac{\tau_{n_0,m_0+1}^{(0)}\tau_{n_0,m_0}^{(1)}}{\tau_{n_0,m_0+1}^{(1)}\tau_{n_0,m_0}^{(0)}}
 \end{pmatrix};
\end{align*}
\item{Two-component system}
\begin{align*}
 &a_1\left(\tau_{n_0,m_0}^{(1)}\tau_{n_0+1,m_0+1}^{(0)}-\tau_{n_0,m_0+1}^{(0)}\tau_{n_0+1,m_0}^{(1)}\right)
 =a_2\left(\tau_{n_0,m_0}^{(1)}\tau_{n_0+1,m_0+1}^{(0)}-\tau_{n_0+1,m_0}^{(0)}\tau_{n_0,m_0+1}^{(1)}\right), \\
 &a_1\left(\tau_{n_0,m_0}^{(0)}\tau_{n_0+1,m_0+1}^{(1)}-\tau_{n_0,m_0+1}^{(1)}\tau_{n_0+1,m_0}^{(0)}\right)
 =a_2\left(\tau_{n_0,m_0}^{(0)}\tau_{n_0+1,m_0+1}^{(1)}-\tau_{n_0+1,m_0}^{(1)}\tau_{n_0,m_0+1}^{(0)}\right);
\end{align*}
\item{Scalar equation (5-point discrete bilinear \ac{KdV} equation.}
\begin{align*}
 (a_1-a_2)^2\tau_{n_0,m_0}^{(0)}\tau_{n_0+2,m_0+2}^{(0)}-(a_1+a_2)^2\tau_{n_0+2,m_0}^{(0)}\tau_{n_0,m_0+2}^{(0)}
 +4a_1a_2\left(\tau_{n_0+1,m_0+1}^{(0)}\right)^2=0;
\end{align*}

\item{Alternative scalar forms (6-point discrete bilinear \ac{KdV} equations)}
\begin{align*}
 &(a_1+a_2)\tau_{n_0,m_0+1}^{(0)}\tau_{n_0+2,m_0}^{(0)}+(a_1-a_2)\tau_{n_0,m_0}^{(0)}\tau_{n_0+2,m_0+1}^{(0)}
 =2a_1\tau_{n_0+1,m_0}^{(0)}\tau_{n_0+1,m_0+1}^{(0)}, \\
 &(a_1+a_2)\tau_{n_0+1,m_0}^{(0)}\tau_{n_0,m_0+2}^{(0)}-(a_1-a_2)\tau_{n_0,m_0}^{(0)}\tau_{n_0+1,m_0+2}^{(0)}
 =2a_2\tau_{n_0,m_0+1}^{(0)}\tau_{n_0+1,m_0+1}.
\end{align*}

\end{itemize}

\end{itemize}

\subsubsection{$(\alpha,\beta)=(0,1)$}

\begin{itemize}

\item{Additive potential}

\begin{itemize}

\item{First integrals}
\begin{align*}
 &\left(a_1+u_{n_0,m_1}^{(1)}-u_{n_0+1,m_1}^{(0)}\right)\left(a_1+u_{n_0,m_1}^{(0)}-u_{n_0+1,m_1}^{(1)}\right)=a_1^2, \\
 &\left(a_2+u_{n_0,m_1}^{(0)}-u_{n_0,m_1+1}^{(0)}\right)\left(a_2+u_{n_0,m_1}^{(1)}-u_{n_0,m_1+1}^{(1)}\right)=a_2^2;
\end{align*}
\item{Lax matrices}
\begin{align*}
 &\bL_{n_0,m_1}=
 \begin{pmatrix}
  a_1+u_{n_0,m_1}^{(1)}-u_{n_0+1,m_1}^{(0)} & \lambda \\
  \lambda & a_1+u_{n_0,m_1}^{(0)}-u_{n_0+1,m_1}^{(1)}
 \end{pmatrix}, \\
 &\bM_{n_0,m_1}=
 \begin{pmatrix}
  \lambda & a_2+u_{n_0,m_1}^{(0)}-u_{n_0,m_1+1}^{(0)} \\
  a_2+u_{n_0,m_1}^{(1)}-u_{n_0,m_1+1}^{(1)} & \lambda
 \end{pmatrix};
\end{align*}
\item{Two-component system}
\begin{align*}
 &\frac{a_1+u_{n_0,m_1+1}^{(1)}-u_{n_0+1,m_1+1}^{(0)}}{a_1+u_{n_0,m_1}^{(0)}-u_{n_0+1,m_1}^{(1)}}
 =\frac{a_2+u_{n_0+1,m_1}^{(0)}-u_{n_0+1,m_1+1}^{(0)}}{a_2+u_{n_0,m_1}^{(1)}-u_{n_0,m_1+1}^{(1)}}, \\
 &\frac{a_1+u_{n_0,m_1+1}^{(0)}-u_{n_0+1,m_1+1}^{(1)}}{a_1+u_{n_0,m_1}^{(1)}-u_{n_0+1,m_1}^{(0)}}
 =\frac{a_2+u_{n_0+1,m_1}^{(1)}-u_{n_0+1,m_1+1}^{(1)}}{a_2+u_{n_0,m_1}^{(0)}-u_{n_0,m_1+1}^{(0)}}.
\end{align*}

\end{itemize}

\item{Quotient potential}

\begin{itemize}

\item{First integrals}
\begin{align*}
 v_{n_0,m_1}^{(0)}v_{n_0,m_1}^{(1)}=1;
\end{align*}
\item{Lax matrices}
\begin{align*}
 \bL_{n_0,m_1}=
 \begin{pmatrix}
  a_1\frac{v_{n_0+1,m_1}^{(0)}}{v_{n_0,m_1}^{(0)}} & \lambda \\
  \lambda & a_1\frac{v_{n_0+1,m_1}^{(1)}}{v_{n_0,m_1}^{(1)}}
 \end{pmatrix}, \quad
 \bM_{n_0,m_1}=
 \begin{pmatrix}
  \lambda & a_2\frac{v_{n_0,m_1+1}^{(0)}}{v_{n_0,m_1}^{(1)}} \\
  a_2\frac{v_{n_0,m_1+1}^{(1)}}{v_{n_0,m_1}^{(0)}} & \lambda
 \end{pmatrix};
\end{align*}
\item{Two-component system (discrete-time \ac{2DTL} of $A_1^{(1)}$-type)}
\begin{align*}
 &a_1\left(\frac{v_{n_0+1,m_1+1}^{(0)}}{v_{n_0,m_1+1}^{(0)}}-\frac{v_{n_0+1,m_1}^{(0)}}{v_{n_0,m_1}^{(0)}}\right)
 =a_2\left(\frac{v_{n_0+1,m_1+1}^{(0)}}{v_{n_0+1,m_1}^{(1)}}-\frac{v_{n_0,m_1+1}^{(1)}}{v_{n_0,m_1}^{(0)}}\right), \\
 &a_1\left(\frac{v_{n_0+1,m_1+1}^{(1)}}{v_{n_0,m_1+1}^{(1)}}-\frac{v_{n_0+1,m_1}^{(1)}}{v_{n_0,m_1}^{(1)}}\right)
 =a_2\left(\frac{v_{n_0+1,m_1+1}^{(1)}}{v_{n_0+1,m_1}^{(0)}}-\frac{v_{n_0,m_1+1}^{(0)}}{v_{n_0,m_1}^{(1)}}\right);
\end{align*}
\item{Scalar equation (discrete \ac{sG} equation)}
\begin{align*}
 a_1\left(v_{n_0+1,m_1}^{(0)}v_{n_0,m_1+1}^{(0)}-v_{n_0,m_1}^{(0)}v_{n_0+1,m_1+1}^{(0)}\right)
 =a_2\left(1-v_{n_0,m_1}^{(0)}v_{n_0+1,m_1}^{(0)}v_{n_0,m_1+1}^{(0)}v_{n_0+1,m_1+1}^{(0)}\right).
\end{align*}

\end{itemize}

\item{Bilinear potential}

\begin{itemize}

\item{Lax matrices}
\begin{align*}
 &\bL_{n_0,m_1}=
 \begin{pmatrix}
  a_1\frac{\tau_{n_0+1,m_1}^{(1)}\tau_{n_0,m_1}^{(0)}}{\tau_{n_0+1,m_1}^{(0)}\tau_{n_0,m_1}^{(1)}} & \lambda \\
  \lambda & a_1\frac{\tau_{n_0+1,m_1}^{(0)}\tau_{n_0,m_1}^{(1)}}{\tau_{n_0+1,m_1}^{(1)}\tau_{n_0,m_1}^{(0)}}
 \end{pmatrix}, \\
 &\bM_{n_0,m_1}=
 \begin{pmatrix}
  \lambda & a_2\frac{\tau_{n_0,m_1+1}^{(1)}\tau_{n_0,m_1}^{(1)}}{\tau_{n_0,m_1+1}^{(0)}\tau_{n_0,m_1}^{(0)}} \\
  a_2\frac{\tau_{n_0,m_1+1}^{(0)}\tau_{n_0,m_1}^{(0)}}{\tau_{n_0,m_1+1}^{(1)}\tau_{n_0,m_1}^{(1)}} & \lambda
 \end{pmatrix};
\end{align*}

\item{Two-component system (discrete bilinear \ac{2DTL} of $A_1^{(1)}$-type})
\begin{align*}
 &a_1\left(\tau_{n_0,m_1}^{(0)}\tau_{n_0+1,m_1+1}^{(0)}-\tau_{n_0,m_1+1}^{(0)}\tau_{n_0+1,m_1}^{(0)}\right)
 =a_2\left(\tau_{n_0,m_1}^{(0)}\tau_{n_0+1,m_1+1}^{(0)}-\tau_{n_0+1,m_1}^{(1)}\tau_{n_0,m_1+1}^{(1)}\right), \\
 &a_1\left(\tau_{n_0,m_1}^{(1)}\tau_{n_0+1,m_1+1}^{(1)}-\tau_{n_0,m_1+1}^{(1)}\tau_{n_0+1,m_1}^{(1)}\right)
 =a_2\left(\tau_{n_0,m_1}^{(1)}\tau_{n_0+1,m_1+1}^{(1)}-\tau_{n_0+1,m_1}^{(0)}\tau_{n_0,m_1+1}^{(0)}\right).
\end{align*}

\end{itemize}

\end{itemize}

\subsubsection{$(\alpha,\beta)=(1,1)$}

\begin{itemize}

\item{Additive potential}

\begin{itemize}

\item{First integrals}
\begin{align*}
 &\left(a_1+u_{n_1,m_1}^{(0)}-u_{n_1+1,m_1}^{(0)}\right)\left(a_1+u_{n_1,m_1}^{(1)}-u_{n_1+1,m_1}^{(1)}\right)=a_1^2, \\
 &\left(a_2+u_{n_1,m_1}^{(0)}-u_{n_1,m_1+1}^{(0)}\right)\left(a_2+u_{n_1,m_1}^{(1)}-u_{n_1,m_1+1}^{(1)}\right)=a_2^2;
\end{align*}
\item{Lax matrices}
\begin{align*}
 &\bL_{n_1,m_1}=
 \begin{pmatrix}
  \lambda & a_1+u_{n_1,m_1}^{(0)}-u_{n_1+1,m_1}^{(0)} \\
  a_1+u_{n_1,m_1}^{(1)}-u_{n_1+1,m_1}^{(1)} & \lambda
 \end{pmatrix}, \\
 &\bM_{n_1,m_1}=
 \begin{pmatrix}
  \lambda & a_2+u_{n_1,m_1}^{(0)}-u_{n_1,m_1+1}^{(0)} \\
  a_2+u_{n_1,m_1}^{(1)}-u_{n_1,m_1+1}^{(1)} & \lambda
 \end{pmatrix};
\end{align*}
\item{Two-component system}
\begin{align*}
 &\frac{a_1+u_{n_1,m_1+1}^{(0)}-u_{n_1+1,m_1+1}^{(0)}}{a_1+u_{n_1,m_1}^{(1)}-u_{n_1+1,m_1}^{(1)}}
 =\frac{a_2+u_{n_1+1,m_1}^{(0)}-u_{n_1+1,m_1+1}^{(0)}}{a_2+u_{n_1,m_1}^{(1)}-u_{n_1,m_1+1}^{(1)}}, \\
 &\frac{a_1+u_{n_1,m_1+1}^{(1)}-u_{n_1+1,m_1+1}^{(1)}}{a_1+u_{n_1,m_1}^{(0)}-u_{n_1+1,m_1}^{(0)}}
 =\frac{a_2+u_{n_1+1,m_1}^{(1)}-u_{n_1+1,m_1+1}^{(1)}}{a_2+u_{n_1,m_1}^{(0)}-u_{n_1,m_1+1}^{(0)}};
\end{align*}
\item{Scalar form (discrete Schwarzian \ac{KdV} equation)}
\begin{align*}
 \frac{\left(a_1+u_{n_1,m_1}^{(0)}-u_{n_1+1,m_1}^{(0)}\right)\left(a_1+u_{n_1,m_1+1}^{(0)}-u_{n_1+1,m_1+1}^{(0)}\right)}
 {\left(a_2+u_{n_1,m_1}^{(0)}-u_{n_1,m_1+1}^{(0)}\right)\left(a_2+u_{n_1+1,m_1}^{(0)}-u_{n_1+1,m_1+1}^{(0)}\right)}
 =\frac{a_1^2}{a_2^2}.
\end{align*}

\end{itemize}

\item{Quotient potential}

\begin{itemize}

\item{First integrals}
\begin{align*}
 v_{n_1,m_1}^{(0)}v_{n_1,m_1}^{(1)}=1;
\end{align*}
\item{Lax matrices}
\begin{align*}
 \bL_{n_1,m_1}=
 \begin{pmatrix}
  \lambda & a_1\frac{v_{n_1+1,m_1}^{(0)}}{v_{n_1,m_1}^{(1)}} \\
  a_1\frac{v_{n_1+1,m_1}^{(1)}}{v_{n_1,m_1}^{(0)}} & \lambda
 \end{pmatrix}, \quad
 \bM_{n_1,m_1}=
 \begin{pmatrix}
  \lambda & a_2\frac{v_{n_1,m_1+1}^{(0)}}{v_{n_1,m_1}^{(1)}} \\
  a_2\frac{v_{n_1,m_1+1}^{(1)}}{v_{n_1,m_1}^{(0)}} & \lambda
 \end{pmatrix};
\end{align*}
\item{Two-component system}
\begin{align*}
 &a_1\left(\frac{v_{n_1+1,m_1+1}^{(0)}}{v_{n_1,m_1+1}^{(1)}}-\frac{v_{n_1+1,m_1}^{(0)}}{v_{n_1,m_1}^{(1)}}\right)
 =a_2\left(\frac{v_{n_1+1,m_1+1}^{(0)}}{v_{n_1+1,m_1}^{(1)}}-\frac{v_{n_1,m_1+1}^{(0)}}{v_{n_1,m_1}^{(1)}}\right), \\
 &a_1\left(\frac{v_{n_1+1,m_1+1}^{(1)}}{v_{n_1,m_1+1}^{(0)}}-\frac{v_{n_1+1,m_1}^{(1)}}{v_{n_1,m_1}^{(0)}}\right)
 =a_2\left(\frac{v_{n_1+1,m_1+1}^{(1)}}{v_{n_1+1,m_1}^{(0)}}-\frac{v_{n_1,m_1+1}^{(1)}}{v_{n_1,m_1}^{(0)}}\right);
\end{align*}
\item{Scalar form (discrete modified \ac{KdV} equation, cf. \eqref{dKdV:v})}
\begin{align*}
 a_1\left(v_{n_1,m_1}^{(0)}v_{n_1+1,m_1}^{(0)}-v_{n_1,m_1+1}^{(0)}v_{n_1+1,m_1+1}^{(0)}\right)
 =a_2\left(v_{n_1,m_1}^{(0)}v_{n_1,m_1+1}^{(0)}-v_{n_1+1,m_1}^{(0)}v_{n_1+1,m_1+1}^{(0)}\right).
\end{align*}

\end{itemize}

\item{Bilinear potential}

\begin{itemize}

\item{Lax matrices}
\begin{align*}
 &\bL_{n_1,m_1}=
 \begin{pmatrix}
  \lambda & a_1\frac{\tau_{n_1+1,m_1}^{(1)}\tau_{n_1,m_1}^{(1)}}{\tau_{n_1+1,m_1}^{(0)}\tau_{n_1,m_1}^{(0)}} \\
  a_1\frac{\tau_{n_1+1,m_1}^{(0)}\tau_{n_1,m_1}^{(0)}}{\tau_{n_1+1,m_1}^{(1)}\tau_{n_1,m_1}^{(1)}} & \lambda
 \end{pmatrix}, \\
 &\bM_{n_1,m_1}=
 \begin{pmatrix}
  \lambda & a_2\frac{\tau_{n_1,m_1+1}^{(1)}\tau_{n_1,m_1}^{(1)}}{\tau_{n_1,m_1+1}^{(0)}\tau_{n_1,m_1}^{(0)}} \\
  a_2\frac{\tau_{n_1,m_1+1}^{(0)}\tau_{n_1,m_1}^{(0)}}{\tau_{n_1,m_1+1}^{(1)}\tau_{n_1,m_1}^{(1)}} & \lambda
 \end{pmatrix};
\end{align*}

\item{Two-component system}
\begin{align*}
 &a_1\left(\tau_{n_1,m_1}^{(1)}\tau_{n_1+1,m_1+1}^{(0)}-\tau_{n_1,m_1+1}^{(1)}\tau_{n_1+1,m_1}^{(0)}\right)
 =a_2\left(\tau_{n_1,m_1}^{(1)}\tau_{n_1+1,m_1+1}^{(0)}-\tau_{n_1+1,m_1}^{(1)}\tau_{n_1,m_1+1}^{(0)}\right), \\
 &a_1\left(\tau_{n_1,m_1}^{(0)}\tau_{n_1+1,m_1+1}^{(1)}-\tau_{n_1,m_1+1}^{(0)}\tau_{n_1+1,m_1}^{(1)}\right)
 =a_2\left(\tau_{n_1,m_1}^{(0)}\tau_{n_1+1,m_1+1}^{(1)}-\tau_{n_1+1,m_1}^{(0)}\tau_{n_1,m_1+1}^{(1)}\right).
\end{align*}

\end{itemize}

\end{itemize}

\subsection{$\cN=3$}

\subsubsection{$(\alpha,\beta)=(0,0)$}

\begin{itemize}

\item{Additive potential}

\begin{itemize}

\item{First integrals}
\begin{align*}
 &\left(a_1+u_{n_0,m_0}^{(1)}-u_{n_0+1,m_0}^{(0)}\right)\left(a_1+u_{n_0,m_0}^{(2)}-u_{n_0+1,m_0}^{(1)}\right)
 \left(a_1+u_{n_0,m_0}^{(0)}-u_{n_0+1,m_0}^{(2)}\right)=a_1^3, \\
 &\left(a_2+u_{n_0,m_0}^{(1)}-u_{n_0,m_0+1}^{(0)}\right)\left(a_2+u_{n_0,m_0}^{(2)}-u_{n_0,m_0+1}^{(1)}\right)
 \left(a_2+u_{n_0,m_0}^{(0)}-u_{n_0,m_0+1}^{(2)}\right)=a_2^3;
\end{align*}
\item{Lax matrices}
\begin{align*}
 &\bL_{n_0,m_0}=
 \begin{pmatrix}
  a_1+u_{n_0,m_0}^{(1)}-u_{n_0+1,m_0}^{(0)} & \lambda & 0 \\
  0 & a_1+u_{n_0,m_0}^{(2)}-u_{n_0+1,m_0}^{(1)} & \lambda \\
  \lambda & 0 & a_1+u_{n_0,m_0}^{(0)}-u_{n_0+1,m_0}^{(2)}
 \end{pmatrix}, \\
 &\bM_{n_0,m_0}=
 \begin{pmatrix}
  a_2+u_{n_0,m_0}^{(1)}-u_{n_0,m_0+1}^{(0)} & \lambda & 0 \\
  0 & a_2+u_{n_0,m_0}^{(2)}-u_{n_0,m_0+1}^{(1)} & \lambda \\
  \lambda & 0 & a_2+u_{n_0,m_0}^{(0)}-u_{n_0,m_0+1}^{(2)}
 \end{pmatrix};
\end{align*}
\item{Three-component system}
\begin{align*}
 &\frac{a_1+u_{n_0,m_0+1}^{(1)}-u_{n_0+1,m_0+1}^{(0)}}{a_1+u_{n_0,m_0}^{(1)}-u_{n_0+1,m_0}^{(0)}}
 =\frac{a_2+u_{n_0+1,m_0}^{(1)}-u_{n_0+1,m_0+1}^{(0)}}{a_2+u_{n_0,m_0}^{(1)}-u_{n_0,m_0+1}^{(0)}}, \\
 &\frac{a_1+u_{n_0,m_0+1}^{(2)}-u_{n_0+1,m_0+1}^{(1)}}{a_1+u_{n_0,m_0}^{(2)}-u_{n_0+1,m_0}^{(1)}}
 =\frac{a_2+u_{n_0+1,m_0}^{(2)}-u_{n_0+1,m_0+1}^{(1)}}{a_2+u_{n_0,m_0}^{(2)}-u_{n_0,m_0+1}^{(1)}}, \\
 &\frac{a_1+u_{n_0,m_0+1}^{(0)}-u_{n_0+1,m_0+1}^{(2)}}{a_1+u_{n_0,m_0}^{(0)}-u_{n_0+1,m_0}^{(2)}}
 =\frac{a_2+u_{n_0+1,m_0}^{(0)}-u_{n_0+1,m_0+1}^{(2)}}{a_2+u_{n_0,m_0}^{(0)}-u_{n_0,m_0+1}^{(2)}};
\end{align*}
\item{Two-component system}
\begin{align*}
 &\frac{a_1+u_{n_0,m_0+1}^{(1)}-u_{n_0+1,m_0+1}^{(0)}}{a_1+u_{n_0,m_0}^{(1)}-u_{n_0+1,m_0}^{(0)}}
 =\frac{a_2+u_{n_0+1,m_0}^{(1)}-u_{n_0+1,m_0+1}^{(0)}}{a_2+u_{n_0,m_0}^{(1)}-u_{n_0,m_0+1}^{(0)}}, \\
 &\frac{a_1^3}{\left(a_1+u_{n_0,m_0}^{(1)}-u_{n_0+1,m_0}^{(0)}\right)}
 -\frac{a_2^3}{\left(a_2+u_{n_0,m_0}^{(1)}-u_{n_0,m_0+1}^{(0)}\right)} \\
 &\qquad =\left(a_1+a_2+u_{n_0,m_0}^{(0)}-u_{n_0+1,m_0+1}^{(1)}\right)
 \left(a_1-a_2+u_{n_0,m_0+1}^{(1)}-u_{n_0+1,m_0}^{(1)}\right);
\end{align*}
\item{Scalar equation (discrete (unmodified) \ac{BSQ} equation)}
\begin{align*}
 &\frac{a_1^3-a_2^3}{u_{n_0+2,m_0}^{(0)}-u_{n_0+1,m_0+1}^{(0)}}-\frac{a_1^3-a_2^3}{u_{n_0+1,m_0+1}^{(0)}-u_{n_0,m_0+2}^{(0)}} \\
 &\qquad =u_{n_0+2,m_0+2}^{(0)}\left(u_{n_0+2,m_0+1}^{(0)}-u_{n_0+1,m_0+2}^{(0)}\right)-u_{n_0+1,m_0}^{(0)}u_{n_0+2,m_0+1}^{(0)} \\
 &\qquad\qquad\qquad +u_{n_0,m_0+1}^{(0)}u_{n_0+1,m_0+2}^{(0)}+u_{n_0,m_0}^{(0)}\left(u_{n_0+1,m_0}^{(0)}-u_{n_0,m_0+1}^{(0)}\right).
\end{align*}

\end{itemize}

\item{Quotient potential}

\begin{itemize}

\item{First integrals}
\begin{align*}
 v_{n_0,m_0}^{(0)}v_{n_0,m_0}^{(1)}v_{n_0,m_0}^{(2)}=1;
\end{align*}
\item{Lax matrices}
\begin{align*}
 &\bL_{n_0,m_0}=
 \begin{pmatrix}
  a_1\frac{v_{n_0+1,m_0}^{(0)}}{v_{n_0,m_0}^{(0)}} & \lambda & 0 \\
  0 & a_1\frac{v_{n_0+1,m_0}^{(1)}}{v_{n_0,m_0}^{(1)}} & \lambda \\
  \lambda & 0 & a_1\frac{v_{n_0+1,m_0}^{(2)}}{v_{n_0,m_0}^{(2)}}
 \end{pmatrix}, \\
 &\bM_{n_0,m_0}=
 \begin{pmatrix}
  a_2\frac{v_{n_0,m_0+1}^{(0)}}{v_{n_0,m_0}^{(0)}} & \lambda & 0 \\
  0 & a_2\frac{v_{n_0,m_0+1}^{(1)}}{v_{n_0,m_0}^{(1)}} & \lambda \\
  \lambda & 0 & a_2\frac{v_{n_0,m_0+1}^{(2)}}{v_{n_0,m_0}^{(2)}}
 \end{pmatrix};
\end{align*}
\item{Three-component system}
\begin{align*}
 &a_1\left(\frac{v_{n_0+1,m_0+1}^{(0)}}{v_{n_0,m_0+1}^{(0)}}-\frac{v_{n_0+1,m_0}^{(1)}}{v_{n_0,m_0}^{(1)}}\right)
 =a_2\left(\frac{v_{n_0+1,m_0+1}^{(0)}}{v_{n_0+1,m_0}^{(0)}}-\frac{v_{n_0,m_0+1}^{(1)}}{v_{n_0,m_0}^{(1)}}\right), \\
 &a_1\left(\frac{v_{n_0+1,m_0+1}^{(1)}}{v_{n_0,m_0+1}^{(1)}}-\frac{v_{n_0+1,m_0}^{(2)}}{v_{n_0,m_0}^{(2)}}\right)
 =a_2\left(\frac{v_{n_0+1,m_0+1}^{(1)}}{v_{n_0+1,m_0}^{(1)}}-\frac{v_{n_0,m_0+1}^{(2)}}{v_{n_0,m_0}^{(2)}}\right), \\
 &a_1\left(\frac{v_{n_0+1,m_0+1}^{(2)}}{v_{n_0,m_0+1}^{(2)}}-\frac{v_{n_0+1,m_0}^{(0)}}{v_{n_0,m_0}^{(0)}}\right)
 =a_2\left(\frac{v_{n_0+1,m_0+1}^{(2)}}{v_{n_0+1,m_0}^{(2)}}-\frac{v_{n_0,m_0+1}^{(0)}}{v_{n_0,m_0}^{(0)}}\right);
\end{align*}
\item{Two-component system}
\begin{align*}
 &a_1\left(\frac{v_{n_0+1,m_0+1}^{(0)}}{v_{n_0,m_0+1}^{(0)}}-\frac{v_{n_0+1,m_0}^{(1)}}{v_{n_0,m_0}^{(1)}}\right)
 =a_2\left(\frac{v_{n_0+1,m_0+1}^{(0)}}{v_{n_0+1,m_0}^{(0)}}-\frac{v_{n_0,m_0+1}^{(1)}}{v_{n_0,m_0}^{(1)}}\right), \\
 &a_1\left(\frac{v_{n_0+1,m_0+1}^{(1)}}{v_{n_0,m_0+1}^{(1)}}
 -\frac{v_{n_0,m_0}^{(0)}v_{n_0,m_0}^{(1)}}{v_{n_0+1,m_0}^{(0)}v_{n_0+1,m_0}^{(1)}}\right)
 =a_2\left(\frac{v_{n_0+1,m_0+1}^{(1)}}{v_{n_0+1,m_0}^{(1)}}
 -\frac{v_{n_0,m_0}^{(0)}v_{n_0,m_0}^{(1)}}{v_{n_0,m_0+1}^{(0)}v_{n_0,m_0+1}^{(1)}}\right);
\end{align*}
\item{Scalar equation (discrete modified \ac{BSQ} equation)}
\begin{align*}
 &\frac{v_{n_0,m_0}^{(0)}}{v_{n_0+1,m_0}^{(0)}}-\frac{v_{n_0,m_0}^{(0)}}{v_{n_0,m_0+1}^{(0)}}+\frac{v_{n_0+2,m_0+1}^{(0)}}{v_{n_0+2,m_0+2}^{(0)}}-\frac{v_{n_0+1,m_0+2}^{(0)}}{v_{n_0+2,m_0+2}^{(0)}} \\
 &\qquad =\left(\frac{a_1^3v_{n_0+1,m_0+1}^{(0)}-a_2^3v_{n_0,m_0+2}^{(0)}}{v_{n_0,m_0+2}^{(0)}-v_{n_0+1,m_0+1}^{(0)}}\right)\frac{v_{n_0+1,m_0+2}^{(0)}}{v_{n_0,m_0+1}^{(0)}} \\
 &\qquad\qquad\qquad -\left(\frac{a_1^3v_{n_0+2,m_0}^{(0)}-a_2^3v_{n_0+1,m_0+1}^{(0)}}{v_{n_0+1,m_0+1}^{(0)}-v_{n_0+2,m_0}^{(0)}}\right)\frac{v_{n_0+2,m_0+1}^{(0)}}{v_{n_0+1,m_0}^{(0)}}.
\end{align*}

\end{itemize}

\item{Bilinear potential}

\begin{itemize}

\item{Lax matrices}
\begin{align*}
 &\bL_{n_0,m_0}=
 \begin{pmatrix}
  a_1\frac{\tau_{n_0+1,m_0}^{(1)}\tau_{n_0,m_0}^{(0)}}{\tau_{n_0+1,m_0}^{(0)}\tau_{n_0,m_0}^{(1)}} & \lambda & 0 \\
  0 & a_1\frac{\tau_{n_0+1,m_0}^{(2)}\tau_{n_0,m_0}^{(1)}}{\tau_{n_0+1,m_0}^{(1)}\tau_{n_0,m_0}^{(2)}} & \lambda \\
  \lambda & 0 &  a_1\frac{\tau_{n_0+1,m_0}^{(0)}\tau_{n_0,m_0}^{(2)}}{\tau_{n_0+1,m_0}^{(2)}\tau_{n_0,m_0}^{(0)}}
 \end{pmatrix}, \\
 &\bM_{n_0,m_0}=
 \begin{pmatrix}
  a_2\frac{\tau_{n_0,m_0+1}^{(1)}\tau_{n_0,m_0}^{(0)}}{\tau_{n_0,m_0+1}^{(0)}\tau_{n_0,m_0}^{(1)}} & \lambda & 0 \\
  0 & a_2\frac{\tau_{n_0,m_0+1}^{(2)}\tau_{n_0,m_0}^{(1)}}{\tau_{n_0,m_0+1}^{(1)}\tau_{n_0,m_0}^{(2)}} & \lambda \\
  \lambda & 0 & a_2\frac{\tau_{n_0,m_0+1}^{(0)}\tau_{n_0,m_0}^{(2)}}{\tau_{n_0,m_0+1}^{(2)}\tau_{n_0,m_0}^{(0)}}
 \end{pmatrix};
\end{align*}
\item{Three-component system}
\begin{align*}
 &a_1\left(\tau_{n_0,m_0}^{(1)}\tau_{n_0+1,m_0+1}^{(0)}-\tau_{n_0,m_0+1}^{(0)}\tau_{n_0+1,m_0}^{(1)}\right)
 =a_2\left(\tau_{n_0,m_0}^{(1)}\tau_{n_0+1,m_0+1}^{(0)}-\tau_{n_0+1,m_0}^{(0)}\tau_{n_0,m_0+1}^{(1)}\right), \\
 &a_1\left(\tau_{n_0,m_0}^{(2)}\tau_{n_0+1,m_0+1}^{(1)}-\tau_{n_0,m_0+1}^{(1)}\tau_{n_0+1,m_0}^{(2)}\right)
 =a_2\left(\tau_{n_0,m_0}^{(2)}\tau_{n_0+1,m_0+1}^{(1)}-\tau_{n_0+1,m_0}^{(1)}\tau_{n_0,m_0+1}^{(2)}\right), \\
 &a_1\left(\tau_{n_0,m_0}^{(0)}\tau_{n_0+1,m_0+1}^{(2)}-\tau_{n_0,m_0+1}^{(2)}\tau_{n_0+1,m_0}^{(0)}\right)
 =a_2\left(\tau_{n_0,m_0}^{(0)}\tau_{n_0+1,m_0+1}^{(2)}-\tau_{n_0+1,m_0}^{(2)}\tau_{n_0,m_0+1}^{(0)}\right);
\end{align*}
\item{Scalar equation (discrete trilinear \ac{BSQ} equation)}
\begin{align*}
 &\left(a_1^2+a_1a_2+a_2^2\right)\left(\tau_{n_0+1,m_0}^{(0)}\tau_{n_0,m_0+2}^{(0)}\tau_{n_0+2,m_0+1}^{(0)}+\tau_{n_0,m_0+1}^{(0)}\tau_{n_0+2,m_0}^{(0)}\tau_{n_0+1,m_0+2}^{(0)}\right) \\
 &\qquad =-(a_1-a_2)^2\tau_{n_0,m_0}^{(0)}\tau_{n_0+1,m_0+1}^{(0)}\tau_{n_0+2,m_0+2}^{(0)} \\
 &\qquad\qquad\qquad +3a_1^2\tau_{n_0+1,m_0}^{(0)}\tau_{n_0+1,m_0+1}^{(0)}\tau_{n_0+1,m_0+2}^{(0)}
 +3a_2^2\tau_{n_0,m_0+1}^{(0)}\tau_{n_0+1,m_0+1}^{(0)}\tau_{n_0+2,m_0+1}^{(0)}.
\end{align*}

\end{itemize}

\end{itemize}

\subsubsection{$(\alpha,\beta)=(0,2)$}

\begin{itemize}

\item{Additive potential}

\begin{itemize}

\item{First integrals}
\begin{align*}
 &\left(a_1+u_{n_0,m_2}^{(1)}-u_{n_0+1,m_2}^{(0)}\right)\left(a_1+u_{n_0,m_2}^{(2)}-u_{n_0+1,m_2}^{(1)}\right)
 \left(a_1+u_{n_0,m_2}^{(0)}-u_{n_0+1,m_2}^{(2)}\right)=a_1^3, \\
 &\left(a_2+u_{n_0,m_2}^{(0)}-u_{n_0,m_2+1}^{(0)}\right)\left(a_2+u_{n_0,m_2}^{(1)}-u_{n_0,m_2+1}^{(1)}\right)
 \left(a_2+u_{n_0,m_2}^{(2)}-u_{n_0,m_2+1}^{(2)}\right)=a_2^3;
\end{align*}
\item{Lax matrices}
\begin{align*}
 &\bL_{n_0,m_2}=
 \begin{pmatrix}
  a_1+u_{n_0,m_2}^{(1)}-u_{n_0+1,m_2}^{(0)} & \lambda & 0 \\
  0 & a_1+u_{n_0,m_2}^{(2)}-u_{n_0+1,m_2}^{(1)} & \lambda \\
  \lambda & 0 & a_1+u_{n_0,m_2}^{(0)}-u_{n_0+1,m_2}^{(2)}
 \end{pmatrix}, \\
 &\bM_{n_0,m_2}=
 \begin{pmatrix}
  \lambda & 0 & a_2+u_{n_0,m_2}^{(0)}-u_{n_0,m_2+1}^{(0)} \\
  a_2+u_{n_0,m_2}^{(1)}-u_{n_0,m_2+1}^{(1)} & \lambda & 0 \\
  0 & a_2+u_{n_0,m_2}^{(2)}-u_{n_0,m_2+1}^{(2)} & \lambda
 \end{pmatrix};
\end{align*}
\item{Three-component system}
\begin{align*}
 &\frac{a_1+u_{n_0,m_2+1}^{(1)}-u_{n_0+1,m_2+1}^{(0)}}{a_1+u_{n_0,m_2}^{(0)}-u_{n_0+1,m_2}^{(2)}}
 =\frac{a_2+u_{n_0+1,m_2}^{(0)}-u_{n_0+1,m_2+1}^{(0)}}{a_2+u_{n_0,m_2}^{(0)}-u_{n_0,m_2+1}^{(0)}}, \\
 &\frac{a_1+u_{n_0,m_2+1}^{(2)}-u_{n_0+1,m_2+1}^{(1)}}{a_1+u_{n_0,m_2}^{(1)}-u_{n_0+1,m_2}^{(0)}}
 =\frac{a_2+u_{n_0+1,m_2}^{(1)}-u_{n_0+1,m_2+1}^{(1)}}{a_2+u_{n_0,m_2}^{(1)}-u_{n_0,m_2+1}^{(1)}}, \\
 &\frac{a_1+u_{n_0,m_2+1}^{(0)}-u_{n_0+1,m_2+1}^{(2)}}{a_1+u_{n_0,m_2}^{(2)}-u_{n_0+1,m_2}^{(1)}}
 =\frac{a_2+u_{n_0+1,m_2}^{(2)}-u_{n_0+1,m_2+1}^{(2)}}{a_2+u_{n_0,m_2}^{(2)}-u_{n_0,m_2+1}^{(2)}}.
\end{align*}
\end{itemize}

\item{Quotient potential}

\begin{itemize}

\item{First integrals}
\begin{align*}
 v_{n_0,m_2}^{(0)}v_{n_0,m_2}^{(1)}v_{n_0,m_2}^{(2)}=1;
\end{align*}
\item{Lax matrices}
\begin{align*}
 &\bL_{n_0,m_2}=
 \begin{pmatrix}
  a_1\frac{v_{n_0+1,m_2}^{(0)}}{v_{n_0,m_2}^{(0)}} & \lambda & 0 \\
  0 & a_1\frac{v_{n_0+1,m_2}^{(1)}}{v_{n_0,m_2}^{(1)}} & \lambda \\
  \lambda & 0 & a_1\frac{v_{n_0+1,m_2}^{(2)}}{v_{n_0,m_2}^{(2)}}
 \end{pmatrix}, \\
 &\bM_{n_0,m_2}=
 \begin{pmatrix}
  \lambda & 0 & a_2\frac{v_{n_0,m_2+1}^{(0)}}{v_{n_0,m_2}^{(2)}} \\
  a_2\frac{v_{n_0,m_2+1}^{(1)}}{v_{n_0,m_2}^{(0)}} & \lambda & 0 \\
  0 & a_2\frac{v_{n_0,m_2+1}^{(2)}}{v_{n_0,m_2}^{(1)}} & \lambda
 \end{pmatrix};
\end{align*}
\item{Three-component system (discrete-time \ac{2DTL} of $A_2^{(1)}$-type)}
\begin{align*}
 &a_1\left(\frac{v_{n_0+1,m_2+1}^{(0)}}{v_{n_0,m_2+1}^{(0)}}-\frac{v_{n_0+1,m_2}^{(0)}}{v_{n_0,m_2}^{(0)}}\right)
 =a_2\left(\frac{v_{n_0+1,m_2+1}^{(0)}}{v_{n_0+1,m_2}^{(2)}}-\frac{v_{n_0,m_2+1}^{(1)}}{v_{n_0,m_2}^{(0)}}\right), \\
 &a_1\left(\frac{v_{n_0+1,m_2+1}^{(1)}}{v_{n_0,m_2+1}^{(1)}}-\frac{v_{n_0+1,m_2}^{(1)}}{v_{n_0,m_2}^{(1)}}\right)
 =a_2\left(\frac{v_{n_0+1,m_2+1}^{(1)}}{v_{n_0+1,m_2}^{(0)}}-\frac{v_{n_0,m_2+1}^{(2)}}{v_{n_0,m_2}^{(1)}}\right), \\
 &a_1\left(\frac{v_{n_0+1,m_2+1}^{(2)}}{v_{n_0,m_2+1}^{(2)}}-\frac{v_{n_0+1,m_2}^{(2)}}{v_{n_0,m_2}^{(2)}}\right)
 =a_2\left(\frac{v_{n_0+1,m_2+1}^{(2)}}{v_{n_0+1,m_2}^{(1)}}-\frac{v_{n_0,m_2+1}^{(0)}}{v_{n_0,m_2}^{(2)}}\right);
\end{align*}
\item{Two-component system}
\begin{align*}
 &a_1\left(\frac{v_{n_0+1,m_2+1}^{(0)}}{v_{n_0,m_2+1}^{(0)}}-\frac{v_{n_0+1,m_2}^{(0)}}{v_{n_0,m_2}^{(0)}}\right)
 =a_2\left(v_{n_0+1,m_2+1}^{(0)}v_{n_0+1,m_2}^{(0)}v_{n_0+1,m_2}^{(1)}-\frac{v_{n_0,m_2+1}^{(1)}}{v_{n_0,m_2}^{(0)}}\right), \\
 &a_1\left(\frac{v_{n_0+1,m_2+1}^{(1)}}{v_{n_0,m_2+1}^{(1)}}-\frac{v_{n_0+1,m_2}^{(1)}}{v_{n_0,m_2}^{(1)}}\right)
 =a_2\left(\frac{v_{n_0+1,m_2+1}^{(1)}}{v_{n_0+1,m_2}^{(0)}}-\frac{1}{v_{n_0,m_2+1}^{(0)}v_{n_0,m_2+1}^{(1)}v_{n_0,m_2}^{(1)}}\right).
\end{align*}
\end{itemize}

\item{Bilinear potential}

\begin{itemize}

\item{Lax matrices}
\begin{align*}
 &\bL_{n_0,m_2}=
 \begin{pmatrix}
  a_1\frac{\tau_{n_0+1,m_2}^{(1)}\tau_{n_0,m_2}^{(0)}}{\tau_{n_0+1,m_2}^{(0)}\tau_{n_0,m_2}^{(1)}} & \lambda & 0 \\
  0 & a_1\frac{\tau_{n_0+1,m_2}^{(2)}\tau_{n_0,m_2}^{(1)}}{\tau_{n_0+1,m_2}^{(1)}\tau_{n_0,m_2}^{(2)}} & \lambda \\
  \lambda & 0 &  a_1\frac{\tau_{n_0+1,m_2}^{(0)}\tau_{n_0,m_2}^{(2)}}{\tau_{n_0+1,m_2}^{(2)}\tau_{n_0,m_2}^{(0)}}
 \end{pmatrix}, \\
 &\bM_{n_0,m_2}=
 \begin{pmatrix}
  \lambda & 0 & a_2\frac{\tau_{n_0,m_2+1}^{(1)}\tau_{n_0,m_2}^{(2)}}{\tau_{n_0,m_2+1}^{(0)}\tau_{n_0,m_2}^{(0)}} \\
  a_2\frac{\tau_{n_0,m_2+1}^{(2)}\tau_{n_0,m_2}^{(0)}}{\tau_{n_0,m_2+1}^{(1)}\tau_{n_0,m_2}^{(1)}} & \lambda & 0 \\
  0 & a_2\frac{\tau_{n_0,m_2+1}^{(0)}\tau_{n_0,m_2}^{(1)}}{\tau_{n_0,m_2+1}^{(2)}\tau_{n_0,m_2}^{(2)}} & \lambda
 \end{pmatrix};
\end{align*}
\item{Three-component system (discrete-time bilinear \ac{2DTL} of $A_2^{(1)}$-type)}
\begin{align*}
 &a_1\left(\tau_{n_0,m_2}^{(0)}\tau_{n_0+1,m_2+1}^{(0)}-\tau_{n_0,m_2+1}^{(0)}\tau_{n_0+1,m_2}^{(0)}\right)
 =a_2\left(\tau_{n_0,m_2}^{(0)}\tau_{n_0+1,m_2+1}^{(0)}-\tau_{n_0+1,m_2}^{(2)}\tau_{n_0,m_2+1}^{(1)}\right), \\
 &a_1\left(\tau_{n_0,m_2}^{(1)}\tau_{n_0+1,m_2+1}^{(1)}-\tau_{n_0,m_2+1}^{(1)}\tau_{n_0+1,m_2}^{(1)}\right)
 =a_2\left(\tau_{n_0,m_2}^{(1)}\tau_{n_0+1,m_2+1}^{(1)}-\tau_{n_0+1,m_2}^{(0)}\tau_{n_0,m_2+1}^{(2)}\right), \\
 &a_1\left(\tau_{n_0,m_2}^{(2)}\tau_{n_0+1,m_2+1}^{(2)}-\tau_{n_0,m_2+1}^{(2)}\tau_{n_0+1,m_2}^{(2)}\right)
 =a_2\left(\tau_{n_0,m_2}^{(2)}\tau_{n_0+1,m_2+1}^{(2)}-\tau_{n_0+1,m_2}^{(1)}\tau_{n_0,m_2+1}^{(0)}\right).
\end{align*}

\end{itemize}

\end{itemize}

\renewcommand{\bibname}{References}
\bibliographystyle{unsrt}
\bibliography{References-ZN}

\begin{thebibliography}{10}

\bibitem{HJN16}
J.~Hietarinta, N.~Joshi, and F.W. Nijhoff.
\newblock {\em {Discrete Systems and Integrability}}.
\newblock Cambridge University Press, Cambridge, 2016.

\bibitem{DIS04}
B.~Grammaticos, Y.~Kosmann-Schwarzbach, and T.~Tamizhmani.
\newblock {\em Discrete Integrable Systems}.
\newblock Springer, Berlin Heidelberg, 2004.

\bibitem{Hirota1981}
R.~Hirota.
\newblock {Discrete analogue of a generalized Toda equation}.
\newblock {\em J. Phys. Soc. Jpn.}, 50:3785--3791, 1981.

\bibitem{Miwa1982}
T.~Miwa.
\newblock {On Hirota's difference equations}.
\newblock {\em Proc. Jpn. Acad.}, 58A:9--12, 1982.

\bibitem{JM82}
M.~Jimbo and T.~Miwa.
\newblock {Soliton equations and fundamental representations of
  $A_{2l}^{(2)}$}.
\newblock {\em Lett. Math. Phys.}, 6:463--469, 1982.

\bibitem{SNZ14}
Y.~Shi, J.J.C. Nimmo, and D.J. Zhang.
\newblock {Darboux and binary Darboux transformations for discrete integrable
  systems I. Discrete potential KdV equation}.
\newblock {\em J. Phys. A: Math. Theor.}, 47:025205, 2014.

\bibitem{SNZ17}
Y.~Shi, J.J.C. Nimmo, and J.X. Zhao.
\newblock {Darboux and binary Darboux transformations for discrete integrable
  systems II. Discrete potential mKdV equation}.
\newblock {\em SIGMA}, 13:036, 2017.

\bibitem{Hone17}
A.Hone, T.~Kouloukas, and C.~Ward.
\newblock On reductions of the {Hirota-Miwa} equation.
\newblock {\em SIGMA 13 (2017), 057, 17 pages}, 2017.

\bibitem{NW01}
F.W. Nijhoff and A.J. Walker.
\newblock {The discrete and continuous Painlev\'e VI hierarchy and the Garnier
  systems}.
\newblock {\em Glasgow Math. J.}, 43A:109--123, 2001.

\bibitem{ABS03}
V.E. Adler, A.I. Bobenko, and Yu.B. Suris.
\newblock {Classification of integrable equations on quad-graphs. The
  consistency approach}.
\newblock {\em Commun. Math. Phys.}, 233:513--543, 2003.

\bibitem{DS97}
A.~Doliwa and P.M. Santini.
\newblock {Multidimensional quadrilateral lattices are integrable}.
\newblock {\em Phys. Lett. A}, 233:365--372, 1997.

\bibitem{NC95}
F.W. Nijhoff and H.W. Capel.
\newblock {The discrete Korteweg--de Vries equation}.
\newblock {\em Acta Appl. Math.}, 39:133--158, 1995.

\bibitem{NPCQ92}
F.W. Nijhoff, V.G. Papageorgiou, H.W. Capel, and G.R.W. Quispel.
\newblock {The lattice Gel'fand--Dikii hierarchy}.
\newblock {\em Inverse Probl.}, 8:597--621, 1992.

\bibitem{ZZN12}
D.J. Zhang, S.L. Zhao, and F.W. Nijhoff.
\newblock {Direct linearisation of extended lattice BSQ systems}.
\newblock {\em Stud. Appl. Math.}, 129:220--248, 2012.

\bibitem{Hie11}
J.~Hietarinta.
\newblock {Boussinesq-like multi-component lattice equations and
  multi-dimensional consistency}.
\newblock {\em J. Phys. A: Math. Theor.}, 44:165204, 2011.

\bibitem{FX17}
A.P. Fordy and P.~Xenitidis.
\newblock {$\mathbb{Z}_n$ graded discrete Lax pairs and integrable difference
  equations}.
\newblock {\em J. Phys. A: Math. Theor.}, 50:165205, 2017.

\bibitem{Mat91}
V.B. Matveev and M.A. Salle.
\newblock {\em {Darboux Transformations and solitons}}.
\newblock Springer-Verlag, Berlin, 1991.

\bibitem{Rog02}
C.~Rogers and W.K. Schief.
\newblock {\em {B\"acklund and Darboux Transformations}}.
\newblock Cambridge University Press, Cambridge, 2002.

\bibitem{WTLS98}
R.~Willox, T.~Tokihiro, I.~Loris, and J.~Satsuma.
\newblock {The fermionic approach to Darboux transformations}.
\newblock {\em Inverse Probl.}, 14:745--762, 1998.

\bibitem{NW97}
J.J.C. Nimmo and R.~Willox.
\newblock {Darboux transformations for the two-dimensional Toda system}.
\newblock {\em Proc. R. Soc. A}, 453:2497--2525, 1997.

\bibitem{Nim97}
J.J.C. Nimmo.
\newblock {Darboux transformations and the discrete KP equation}.
\newblock {\em J. Phys. A: Math. Gen.}, 30:8693--8704, 1997.

\bibitem{WTS97}
R.~Willox, T.~Tokihiro, and J.~Satsuma.
\newblock {Darboux and binary Darboux transformations for the nonautonomous
  discrete KP equation}.
\newblock {\em J. Math. Phys.}, 38:6455--6469, 1997.

\bibitem{DN09}
A.~Doliwa and M.~Nieszporski.
\newblock {Darboux transformations for linear operators on two-dimensional
  regular lattices}.
\newblock {\em J. Phys. A: Math. Theor.}, 42:454001, 2009.

\bibitem{DN91}
I.Ya. Dorfman and F.W. Nijhoff.
\newblock {On a (2+1)-dimensional version of the Krichever--Novikov equation}.
\newblock {\em Phys. Lett. A}, 157:107--112, 1991.

\bibitem{NCWQ84}
F.W. Nijhoff, H.W. Capel, G.L. Wiersma, and G.R.W. Quispel.
\newblock {B\"acklund transformations and three-dimensional lattice equations}.
\newblock {\em Phys. Lett. A}, 105:267--272, 1984.

\bibitem{Nim06}
J.J.C. Nimmo.
\newblock {On a non-Abelian Hirota--Miwa equation}.
\newblock {\em J. Phys. A: Math. Gen.}, 39:5053--5065, 2006.

\bibitem{GRPSW07}
B.~Grammaticos, A.~Ramani, V.~Papageorgiou, J.~Satsuma, and R.~Willox.
\newblock {Constructing lump-like solutions of the Hirota--Miwa equation}.
\newblock {\em J. Phys. A: Math. Theor.}, 40:12619--12627, 2007.

\bibitem{Fu17a}
W.~Fu and F.W. Nijhoff.
\newblock {Direct linearizing transform for three-dimensional discrete
  integrable systems: the lattice AKP, BKP and CKP equations}.
\newblock {\em Proc. R. Soc. A}, 473:20160915, 2017.

\bibitem{NCW85}
F.W. Nijhoff, H.W. Capel, and G.L. Wiersma.
\newblock {Integrable lattice systems in two and three dimensions}.
\newblock In R.~Martini, editor, {\em Geometric Aspects of the Einstein
  Equations and Integrable Systems}, volume 239 of {\em Lecture Notes in
  Physics}, pages 263--302. Springer-Verlag, Berlin, 1985.

\bibitem{Fu17b}
W.~Fu and F.W. Nijhoff.
\newblock {On reductions of the discrete Kadomtsev--Petviashvili-type
  equations}.
\newblock {\em J. Phys. A: Math. Theor.}, 50:505203, 2017.

\bibitem{Fu18a}
W.~Fu.
\newblock {Direct linearisation of the discrete-time two-dimensional Toda
  lattices}.
\newblock {\em J. Phys. A: Math. Theor.}, 51:334001, 2018.

\bibitem{Fu18b}
W.~Fu and F.W. Nijhoff.
\newblock {Linear integral equations, infinite matrices and soliton
  hierarchies}.
\newblock {\em J. Math. Phys.}, 59:071101, 2018.

\bibitem{JM83}
M.~Jimbo and T.~Miwa.
\newblock {Solitons and infinite dimensional Lie algebras}.
\newblock {\em Publ. RIMS}, 19:943--1001, 1983.

\bibitem{ALN12}
J.~Atkinson, S.B. Lobb, and F.W. Nijhoff.
\newblock {An integrable multicomponent quad equation and its Lagrangian
  formulation}.
\newblock {\em Theor. Math. Phys.}, 173:1644--1653, 2012.

\bibitem{FG80}
A.P. Fordy and J.~Gibbons.
\newblock {Integrable nonlinear Klein--Gordon equations and Toda lattices}.
\newblock {\em Commun. Math. Phys.}, 77:21--30, 1980.

\end{thebibliography}
\end{document}